\documentclass[11pt]{article}
\usepackage{paralist}
\usepackage{color}
\usepackage{bm}

\usepackage{soul}

\usepackage[cm]{fullpage}
\usepackage{amsfonts}
\usepackage{amssymb}

\usepackage{amsmath,enumitem}
\usepackage[usenames,dvipsnames]{xcolor}
\usepackage{tikz,pgf,color}
\usepackage{constants}
\usepackage{amsthm}
\makeatletter
\newtheorem*{rep@theorem}{\rep@title}
\newcommand{\newreptheorem}[2]{%
	\newenvironment{rep#1}[1]{%
		\def\rep@title{#2 \ref{##1}}%
		\begin{rep@theorem}}%
		{\end{rep@theorem}}}
\makeatother

\usepackage{ifpdf}
\newcommand{\mydriver}{hypertex}
\ifpdf
\renewcommand{\mydriver}{pdftex}
\fi
\usepackage[breaklinks,\mydriver]{hyperref}

\usepackage[margin=1in]{geometry}

\theoremstyle{plain}
\newtheorem{theorem}{Theorem}[section]
\newtheorem{fact}[theorem]{Fact}
\newtheorem{lemma}[theorem]{Lemma}

\newtheorem{conjecture}[theorem]{Conjecture}

\newtheorem{definition}[theorem]{Definition}

\usepackage{algorithm}
\usepackage{algpseudocode}
\usepackage[normalem]{ulem}

\algtext*{EndWhile}
\algtext*{EndIf}
\algtext*{EndFor}

\newcommand{\region}{\ensuremath{\mathcal{P}}}

\newcommand{\asc}{\mathrm{asc}}

\newcommand{\divi}{\textrm{div}}
\newcommand{\Ds}{\ensuremath{\mathcal{D}}}
\newcommand{\re}{\ensuremath{\mathrm{re}}}

\DeclareMathOperator{\mincut}{mincut}
\DeclareMathOperator{\capacity}{cap}

\newcommand{\energy}{\ensuremath{\mathcal{E}}}

\newcommand{\poly}{\ensuremath{\mathrm{poly}}}
\newcommand{\vect}[1]{\ensuremath{\mathbf{#1}}}
\newcommand{\mat}[1]{\ensuremath{\mathbf{#1}}}

\newcommand{\real}{\mathbb{R}}

\newcommand{\val}{\mathrm{val}}

\title{The Power of Vertex Sparsifiers in Dynamic Graph Algorithms\footnote{The research leading to these results has received funding from the
		European Research Council under the European Union's Seventh	Framework Programme (FP/2007-2013) / ERC Grant Agreement no. 340506.}}
\author{Gramoz Goranci\footnote{University of Vienna, Faculty of Computer Science, Vienna, Austria. E-mail: \texttt{gramoz.goranci@univie.ac.at}.}
	\and
	Monika Henzinger\footnote{University of Vienna, Faculty of Computer Science, Vienna, Austria. E-mail: \texttt{monika.henzinger@univie.ac.at}.}
	\and 
	Pan Peng\footnote{University of Vienna, Faculty of Computer Science, Vienna, Austria. E-mail: \texttt{pan.peng@univie.ac.at}.}
}
\date{}

\begin{document}

\begin{titlepage}
	\maketitle
	\thispagestyle{empty}

\begin{abstract}
We introduce a new algorithmic framework for designing dynamic graph algorithms in minor-free graphs, by exploiting the structure of such graphs and a tool called \emph{vertex sparsification}, which is a way to compress large graphs into small ones that well preserve relevant properties among a subset of vertices and has previously mainly been used in the design of approximation algorithms. 

Using this framework, we obtain a Monte Carlo randomized fully dynamic algorithm for $(1+\varepsilon)$-approximating the energy of electrical flows in $n$-vertex planar graphs with $\tilde{O}(r\varepsilon^{-2})$ worst-case update time and $\tilde{O}((r+\frac{n}{\sqrt{r}})\varepsilon^{-2})$ worst-case query time, for any $r$ larger than some constant. For $r=n^{2/3}$, this gives $\tilde{O}(n^{2/3}\varepsilon^{-2})$ update time and $\tilde{O}(n^{2/3}\varepsilon^{-2})$ query time. We also extend this algorithm to work for minor-free graphs with similar approximation and running time guarantees. Furthermore, we illustrate our framework on the all-pairs max flow and shortest path problems by giving corresponding dynamic algorithms in minor-free graphs with both sublinear update and query times. To the best of our knowledge, our results are the first to systematically establish such a connection between dynamic graph algorithms and vertex sparsification. 

We also present both upper bound and lower bound for maintaining the energy of electrical flows in the incremental subgraph model, where updates consist of only vertex activations, which might be of independent interest.
 \end{abstract}
\end{titlepage}
 
\section{Introduction}
A \emph{dynamic graph} is a graph that undergoes constant changes over time. Such changes or updates may correspond to inserting/deleting an edge or activating/deactivating a vertex from the graph. The goal of a  \emph{dynamic graph algorithm} is to maintain some property of a graph and support an intermixed sequence of update and query operations that can be processed quickly. In particular, the algorithm should at least beat the trivial one that recomputes the solution from scratch after each update. The last three decades have witnessed a large body of research on dynamic graph algorithms for a number of fundamental properties, including connectivity, minimum spanning tree, shortest path, matching and so on. Most of these problems have been considered in both general graphs as well as planar graphs, with quite different techniques and trade-offs between update and query times. In particular, many dynamic algorithms for planar graphs heavily depend on the duality of planar graphs~\cite{Subramanian93,KleinS98,ItalianoNSW11} and do not seem easily generalizable to a larger class of graphs, e.g., the family of minor-free graphs.

In this paper, we provide a new algorithmic framework for designing dynamic graph algorithms in \emph{minor-free} graphs that are free of a $K_t$-minor for any fixed integer $t\geq 1$, by utilizing a tool called \emph{vertex sparsification} as well as the structure of minor-free graphs. Vertex sparsification is a way of compressing large graphs into smaller ones that well preserve the relevant properties (e.g., cut, flow and distance information) among a subset of vertices (called \emph{terminals})~(e.g., \cite{Moitra09,andoni2014towards,KNZ14preserving}). Besides the natural motivation of achieving more space-efficient storage and obtaining faster algorithms on the reduced graphs, it has also found applications in the design of approximation algorithms~\cite{Moitra09}, network design and routing~\cite{chuzhoy2012routing}. We show that good quality and efficiently constructible vertex sparsifiers can be used to give efficient dynamic graph algorithms. To the best of our knowledge, our results are the first to systematically establish such a connection between dynamic graph algorithms and vertex sparsification. 

We illustrate our algorithmic framework on the \emph{all-pairs electrical flow, all-pairs max flow} and \emph{all-pairs shortest path} problems in minor-free graphs. Previously, there is no known dynamic algorithms for the first problem (even for special class of graphs), and for the second problem, we only know dynamic algorithms for planar graphs. 

%


%

\subsection{Our Results}
\subsubsection{Electrical Flow}

The \emph{electrical flow} problem is one of the most fundamental problems in electrical engineering and physics~\cite{doyle84}, and recently received increasing interest in computer science due to its close relation to linear equation solvers~\cite{SpielmanT14,KoutisMP14}, graph sparsification~\cite{SpielmanT04,SpielmanS11}, maximum flows (and minimum cuts)~\cite{ChristianoKMST11,LeeRS13,madry2013,KelnerLOS14,Madry16}. Slightly more formally, the $s-t$ electrical flow problem asks to find the flow (current) that minimizes the energy dissipation of a weighted graph when one unit of flow is injected at the source $s$ and extracted at the sink $t$.  


In the dynamic \emph{all-pairs electrical flow} problem, our objective is to minimize the update time and the query time for outputting the exact (or approximate) energy of the $s-t$ electrical flow (see Section~\ref{sec:preliminaries} for formal definitions) in the current graph, for any two vertices $s,t$. 
In the following, we will focus on \emph{fully dynamic} graphs, in which updates consist of both edge insertions and deletions. 
Specifically, we allow the following operations:
\begin{itemize}
	\itemsep0.1em 
	\item \textsc{Insert}$(u,v,r)$: Insert the edge $(u,v)$ with resistance $r$ in $G$, provided that the new edge preserves the planarity (or minor-freeness) of $G$.
	\item \textsc{Delete}$(u,v)$: Delete the edge $(u,v)$ from $G$.
	\item \textsc{ElectricalFlow}$(s,t)$: Return the exact (or approximate) energy of the $s-t$ electrical flow in the current graph $G$. 
\end{itemize}
For a graph $G$ and two vertices $s,t$, we let $\energy_G(s,t)$ denote the energy of the $s-t$ electrical flow. For any $\alpha \geq 1$, we say that an algorithm is an $\alpha$-approximation to $\energy_G(s,t)$ if \textsc{ElectricalFlow}$(s,t)$ returns a positive number $k$ such that $\energy_G(s,t) \leq k \leq \alpha \cdot \energy_G(s,t)$. 

We present the first non-trivial fully dynamic algorithm for maintaining a $(1+\varepsilon)$-approximation to the energy of the $s-t$ electrical flow in planar and minor-free graphs. Our algorithm achieves both sublinear worst-case query and update times\footnote{Throughout the paper, we use $\tilde{O}(\cdot)$ to hide polylogarithmic factors, i.e., $\tilde{O}(f(n))=O(f(n)\cdot\poly\log f(n))$; ``with high probability'' refers to ``with probability at least $1-\frac{1}{n^{c}}$, for some $c>0$''.}.

\begin{theorem}~\label{thm:edge_worst} Fix $\varepsilon \in (0,1)$ and integer $t>0$. Let $r\geq c$ for some large constant $c>0$. Given a $K_t$-minor-free graph $G=(V,E,\vect{w})$ with positive edge weights, we can maintain a $(1+\varepsilon)$-approximation to the all-pairs electrical flow problem with high probability. The worst-case update time per operation is $\tilde{O}(\frac{n^\xi \cdot r}{\varepsilon^{2}})$ and the worst-case query time is $\tilde{O}((r+n/\sqrt{r})\varepsilon^{-2})$, for any constant $\xi>0$. Furthermore, if $G$ is planar, then $\xi$ can be chosen to be $0$.
\end{theorem}

Note that by setting $r=n^{2/3}$, we obtain a dynamic algorithm for planar graphs with worst-case 
$\tilde{O}(n^{2/3}\varepsilon^{-2})$ update time and $\tilde{O}(n^{2/3}\varepsilon^{-2})$ query time. One may be tempted to reduce our problem to dynamically maintaining \emph{spectral sparsifiers}.
Despite the fact that such sparsifiers approximately preserve electrical flows and that a $(1\pm\varepsilon)$-spectral sparsifier can be maintained with amortized update time $\poly(\log n,\varepsilon^{-1})$~\cite{AbrahamDKKP16}, performing query operations on the sparsifier about the energy of $s-t$ electrical flow still requires $\Omega(n)$ time. 

We also give a dynamic algorithm for all-pairs electrical flow for minor-free graphs in the \emph{incremental subgraph} model, where the updates in the dynamic graph are a sequence of \emph{vertex activation} operations. Our algorithm maintains a $(1+\varepsilon)$-approximation of the energy of electrical flows in minor-free graphs with $\tilde{O}(r\varepsilon^{-2})$ \emph{amortized} update time and $\tilde{O}((r+n/\sqrt{r})\varepsilon^{-2})$ \emph{worst-case} query time, for any $r\geq c$ (see Theorem~\ref{thm:alg_subgraph}). For $r=n^{2/3}$, this gives $\tilde{O}(n^{2/3}\varepsilon^{-2})$ amortized update and worst-case query time. We complement this result by showing the following conditional lower bound (see Theorem~\ref{thm:subgraph_lower}): there is no incremental algorithm in the subgraph model that $C$-approximates the energy of the electrical flows in \emph{general} graphs with both $O(n^{1-\varepsilon})$ worst-case update time and $O(n^{2-\varepsilon})$ worst-case query time, for any $C>0$ and constant $\varepsilon>0$, unless the \emph{online matrix vector multiplication (oMv)} conjecture if false. Our results show a polynomial gap of dynamic algorithms for subgraph electrical flows between minor-free graphs and general graphs, conditioned on the oMv conjecture. These results might be of independent interest.

%



\subsubsection{Max Flow}
Our second result from our algorithmic framework is a dynamic algorithm for all-pairs max flows in minor free graphs. In this problem, the updates consist of both edge insertions and deletions similar to the case for the electrical flow problem, while an edge insertion here corresponds to inserting an edge with some weight rather than resistance. The query \textsc{MaxFlow}$(s,t)$ asks the exact (or approximate) $s-t$ max flow value in the current graph.
\begin{theorem}~\label{thm:maxflow_1}
	Let $t>0$ be a fixed integer. Let $r\geq c$ for some large constant $c>0$. Given a $K_t$-minor-free graph $G=(V,E,\vect{w})$ with positive edge weights, we can maintain a $O(1)$-approximation to the all-pairs max-flows problem. The worst-case update time is $\tilde{O}(n^\xi r+r^3)$ for any constant $\xi>0$, and the worst-case query time is $\tilde{O}(r+n/\sqrt{r})$. One can also maintain a $O(\log^4 n)$-approximation to the all-pairs max-flows problem, with worst-case update time $\tilde{O}(n^\xi r+r)$, and worst-case query time $\tilde{O}(r+n/\sqrt{r})$.
\end{theorem}

Note that by setting $r\in \left(\poly\log n, o(n^{1/3})\right)$, we can maintain a $O(1)$-approximation to the all-pairs max flows problem in minor free graphs with both sublinear worst-case update and query times. By setting $r=n^{2/3}$, we can obtain $O(\log^4 n)$-approximation with worst-case $\tilde{O}(n^{\xi+2/3})$ update time and $\tilde{O}(n^{2/3})$ query time. 

We remark that Italiano et al.~\cite{ItalianoNSW11} have given a fully dynamic algorithm for \emph{exact} all-pairs max-flow in planar graphs with worst-case $\tilde{O}(n^{2/3})$ update and $\tilde{O}(n^{2/3})$ query time. Their algorithm is based on maintaining an edge decomposition (called \emph{$r$-division}) of the planar graph, which is similar to ours, while there are some substantial differences. First of all, their algorithm does not seem easily generalizable to minor-free graphs since it depends on the duality of planar graphs. Second, it is required that the embedding of the graph does not change throughout the sequence of updates~\cite{ItalianoNSW11}, which is not necessary in our algorithm. Third, though their algorithm can answer the exact max flow value with the aforementioned running time guarantee, it does not provide an update/query trade-off as ours. 

\subsubsection{Shortest Path}
Our third result is a fully dynamic algorithm for all-pairs shortest paths in minor-free graphs. In this problem, the updates consist of edge insertions and deletions as above and the query \textsc{ShortestPath}$(s,t)$ asks the exact (or approximate) shortest path length between $s$ and $t$ in the current graph.
\begin{theorem}~\label{thm:shortest_path} Let $t,q\geq 1$. Given a $K_t$-minor-free graph $G=(V,E,\vect{w})$ with positive edge weights, we can maintain a $(2q-1)$-approximation to the all-pair shortest path problem. The worst-case expected update time is $\tilde{O}(n^{6/7})$ and the worst-case query time is $\tilde{O}(n^{\frac{6}{7}+\frac{3}{7q}})$.
\end{theorem}

Note that for $q\geq 4$, both update time and query time of the above algorithm will be sublinear. We remark that for the special case of planar graphs, the above running time and approximation guarantee are worse than the result of Abraham et al.~\cite{AbrahamCG12}, who gave a fully dynamic $(1+\varepsilon)$-approximation for planar shortest path problem with worst-case $\tilde{O}(\sqrt{n})$ update and query time. 
However, it is unclear how to generalize their algorithm to minor-free graphs. There are also works on fully dynamic all-pairs shortest path in general graphs~(e.g., \cite{Ber09:fully,ACT14:fully}), for which there is no known algorithm with non-trivial worst-case update time that breaks $O(n)$ barrier.

\paragraph{Remark.} We want to point it out that all the above results might be generalized to a larger class of graphs that admit efficiently constructible good separators, while our main focus is to bring up this new algorithmic framework. 

\subsection{Our Techniques}
Our fully dynamic algorithms in planar and minor-free graphs combine the ideas of maintaining an edge decomposition of the current graph $G$ and approximately preserving the relevant properties or quantities by smaller ``substitutes'', which allow us to operate on a small piece of the graph during each update (in the amortized sense) and significantly reduce the size of the query graph that well preserves the property of $G$. These ``substitutes'' refer to the vertex sparsifiers for the corresponding properties. 

Such an edge decomposition is called $r$-division~\cite{Fed87:shortest}. Given some graph $G$ and a parameter $r$, we partition $G$ into a collection of  $O(n/r)$ edge disjoint subgraphs (called \emph{regions}), each contains at most $O(r)$ vertices. This induces a partitioning of the vertex set into \emph{interior} vertices (those that are incident only to vertices within the same region) and \emph{boundary} vertices (those that are incident to vertices in different regions). In addition, we ensure that the total number of boundary vertices is $O(n/\sqrt{r})$. 
Maintaining an $r$-division has also been used in some previous dynamic algorithms for planar graphs~\cite{Subramanian93, KleinS98, GalilIS99, ItalianoNSW11}.

Now a key observation is that for any $s,t$, by removing all the interior vertices from other regions that do not contain $s,t$ and adding some edges with appropriate weights among boundary vertices, one can guarantee that the resulting graph exactly preserves the quantities between $s,t$ (e.g., the energy of $s-t$ electrical flow, the value of $s-t$ max flow). Now let us elaborate on the electrical flow problem, for which the aforementioned reduction is called \emph{Schur Complement}. The problem of performing such a Schur Complement on a region is that it is very time-consuming as it adds too many edges among boundary vertices. Instead, we resort to a recent tool called approximate Schur Complement~(\cite{DurfeeKPRS16}; see Section \ref{sec: SchurComplement}), which well approximates the pairwise effective resistances among boundary vertices and also gives a sparse graph (or a substitute) induced by all boundary vertices. Now for an update, we only recompute a constant number of such substitutes (and we need to periodically rebuild the data structure); for a query, we take the small graph defined by choosing appropriate regions and substitutes, and answer the query according to the $s-t$ effective resistance on this small graph. Since such a substitute can be computed very fast and is sparse, we are ensured to obtain sublinear amortized update time and worst-case query time. Using a global rebuilding technique, we show that one can also achieve worst-case update time.

Our approach differs from the previous dynamic planar graph algorithms in that the $r$-division we use does not require that the boundary of each region contains a constant number of faces or the duality of planar graphs, since we only need to maintain the $r$-division and fast compute the approximate Schur Complement. 

Such an approximate Schur Complement can be viewed as a vertex \emph{spectral/resistance sparsifier} by treating boundary vertices as terminals. To obtain dynamic algorithms for the all-pairs max flows (resp., shortest paths) problems, we can use vertex cut sparsifiers (resp., distance sparsifiers), which well preserve the values of minimum cut separating any subset of terminals (resp., the distances among all terminal pairs).


%
%



\subsection{Related Work}
In the static setting, the electrical flow problem amounts to solving a system of linear equations, where the underlying matrix is a Laplacian (see the monograph of Doyle and Snell~\cite{doyle84}). Christiano et al.~\cite{ChristianoKMST11} used the electrical flow computation as a subroutine within the multiplicative-weights update framework~\cite{AroraHK12}, to obtain the breakthrough result of $(1-\varepsilon)$-approximating the undirected maximum $s-t$ flow (and the minimum $s-t$ cut) in $\tilde{O}(mn^{1/3}\varepsilon^{-11/3})$ time.
This has inspired and led to further development of fast algorithms for approximating $s-t$ maximum flow, which culminated in an $\tilde{O}(m)$ time algorithm for this problem in undirected graphs~\cite{Peng16}.

Lipton, Rose and Tarjan~\cite{LiptonRT79} consider the problem of designing fast algorithms for \emph{exactly} solving linear systems where the matrix is positive definite and the associated graph is planar. Their result implies an $O(n^{3/2})$ time algorithm for electrical flow in planar graphs. This was latter improved to $O(n^{\omega/2})$ by Alon and Yuster~\cite{AlonY10}, where $\omega=2.37..$ is the exponent in the running time of the fastest algorithm for matrix multiplication~\cite{Williams12}. Miller and Koutis~\cite{KoutisM07} consider parallel algorithms for approximately solving planar Laplacian systems. Their algorithm runs in $\tilde{O}(n^{1/6 +c})$ parallel time and $O(n)$ work, where $c$ is any positive constant. We refer the reader to \cite{kenyon2011} for other useful properties of Laplacians on planar graphs. 

Related data structure concepts dealing with spectral properties of graphs include semi-streaming and dynamic algorithms for maintaining spectral sparsifiers. Kelner and Levin~\cite{KelnerL13} give single-pass incremental streaming algorithm using near-linear space and total update time. This was extended by Kapralov et al.~\cite{KapralovLMMS14} to the dynamic semi-streaming model which allows both edge insertions and deletions. Recently, Abraham et al.~\cite{AbrahamDKKP16} give a fully-dynamic algorithm for maintaining spectral sparsifiers in poly-logarithmic amortized update time.

There is a line of work on dynamic algorithms for planar graphs that maintains information about important measures like reachability, connectivity, shortest path, max-flow etc. Subramanian~\cite{Subramanian93} shows a fully-dynamic algorithm for maintaining reachability in directed planar graphs in $O(n^{2/3} \log n)$ time per operation. For the connectivity measure, Eppstein et al.~\cite{EppsteinGIS96} give an algorithm with $O(\log^{2} n)$ amortized update time and $O(\log n)$ query time. Dynamic all-pairs shortest path problem in planar graphs was initiated by Klein and Subramanian~\cite{KleinS98}, who showed how to maintain a $(1+\varepsilon)$-approximation to shortest paths in $O(n^{2/3} \log^{2} n \log D)$ amortized update time and $O(n^{2/3} \log^{2} n \log D)$ worst-case query time, where $D$ denotes the sum of edge lengths. The best known algorithm is due to Abraham, Chechik and Gavoille~\cite{AbrahamCG12} and maintains a $(1+\varepsilon)$-approximation in $O(\sqrt{n} \log^{2} n/ \varepsilon)$ worst-case time per operation. Italiano et al.~\cite{ItalianoNSW11} obtain a fully-dynamic algorithm for \emph{exact} $s-t$ max-flow in planar graphs with $O(n^{2/3} \log^{8/3})$ worst-case time per operation.

Motivated by the recent developments on proving conditional lower-bounds for dynamic problems~\cite{AbboudW14,henzinger2015unifying}, Abboud and Dahlgaard~\cite{AbboudD16} give conditional lower-bounds for a class of dynamic graph problems restricted to planar graphs. Specifically, under the conjecture that all-pair-shortest path problem cannot be solved in truly subcubic time, they show that no algorithm for dynamic shortest path in planar graphs can support both updates and queries in $O(n^{1/2 - \varepsilon})$ amortized time, for $\varepsilon > 0$. 
\section{Preliminaries}~\label{sec:preliminaries}
We consider a weighted undirected graph $G$ undergoing edge insertions/deletions or vertex activations/deactivations. Our dynamic algorithms are characterized by two time measures: \emph{query time}, which denotes the time needed to answer a query, and \emph{update time}, which denotes the time needed to perform an update operation. We say that an algorithm has $O(t(n))$ \emph{worst-case} update time, if it takes $O(t(n))$ time to process \emph{each} update. We say that an algorithm has $O(t(n))$ \emph{amortized} update time if it takes $O(f \cdot t(n))$ total update time for processing $f$ updates (edge insertions/deletions or vertex activations/deactivations). 
\paragraph{Basic Definitions.} 
Let $G=(V,E,\vect{w})$ be any undirected weighted graph with $n$ vertices and $m$ edges, where for any edge $e$, its weight $\vect{w}(e)>0$. 
Let $\mat{A}$ denote the {weighted adjacency matrix}, let $\mat{D}$ denote the weighted degree diagonal matrix, and let $\mat{L}=\mat{D}-\mat{A}$ denote the \emph{Laplacian} matrix of $G$. We fix an arbitrary orientation of edges, that is, for any two vertices $u,v$ connected by an edge, exactly one of $(u,v)\in E$ or $(v,u)\in E$ holds. Let $\mat{B}\in \real^{m\times n}$ denote the incidence matrix of $G$ such that for any edge $e=(u,v)$ and vertex $w\in V$, $\mat{B}((u,v),w) = 1$ if $u=w$, $-1$ if $v=w$, and $0$ otherwise. We will also think of the weight $\vect{w}(e)$ of any edge $e$ as the \emph{conductance} of $e$, and its reciprocal $\frac{1}{\vect{w}(e)}$, denoted as $\vect{r}(e)$, as the \emph{resistance} of $e$. Let $\mat{R}\in \real^{m\times m}$ denote a diagonal matrix with $\mat{R}(e,e)=\vect{r}(e)$, for any edge $e$. Note that $\mat{L}=\mat{B}^T\mat{R}^{-1}\mat{B}$. 

For any $\vect{x} \in \mathbb{R}^{n}$, the quadratic form associated with $\mat{L}$ is given by $\vect{x}^{T}\mat{L}\vect{x}$. For any two different vertices $u,v$, let $\vect{\chi}_{u,v}\in \real^n$ denote the vector such that $\vect{\chi}_{u,v}(w)=1$ if $w=u$, $-1$ if $w=v$ and $0$ otherwise. 
For any two vertices $s,t\in V$, an \emph{$s-t$ flow} is a mapping $\vect{f}:E\rightarrow \mathbb{R}^+$ satisfying the following conservation constraint: for any $v\neq s,t$, it holds that 
$\sum_{e= (v,u)} \vect{f}(e) = \sum_{e=(u,v)} \vect{f}(e),$
where for any edge $e=(v,u)$, $\vect{f}(e):=\vect{f}(v,u)$ and $\vect{f}(u,v):=-\vect{f}(v,u)$.

We will let $\val(\vect{f})=\sum_{v:(s,v)\in E} \vect{f}(s,v)$ denote the \emph{value of an $s-t$ flow}. Note that for an $s-t$ flow with value $1$, it holds that $\mat{B}^T\vect{f}=\vect{\chi}_{s,t}$. Given an $s-t$ flow $\vect{f}$, its \emph{energy} (with respect to the resistance vector $\vect{r}$) is defined as $\energy_{\vect{r}}(\vect{f}, s,t)=\sum_{e}\vect{r}(e)\vect{f}(e)^2=\vect{f}^T\mat{R}\vect{f}.$

We define the \emph{$s-t$ electrical flow} in $G$ to be the $s-t$ flow that minimizes the energy $\energy_{\vect{r}}(\vect{f},s,t)$ among all $s-t$ flows with \emph{unit} flow value. It is known that such a flow is unique~\cite{doyle84}.

Any $s-t$ flow $\vect{f}$ in $G$ is an $s-t$ electrical flow with respect to $\vect{r}$, iff there exists a vertex potential function $\vect{\phi}:V\rightarrow \real^+$ such that for any $e=(u,v)$ that is oriented from $u$ to $v$, $\vect{f}(e)=\frac{\vect{\phi}(v)-\vect{\phi}(u)}{\vect{r}(e)}$. It is known that such a vector $\vect{\phi}$ satisfies that $\vect{\phi}=\mat{L}^\dagger \vect{\chi}_{s,t}$, where $\mat{L}^\dagger$ denotes the (Moore-Penrose) pseudo-inverse of $\mat{L}$. In addition,  $\vect{f}=\mat{R}^{-1}\mat{B}^T\vect{\phi}=\mat{R}^{-1}\mat
{B}^T\mat{L}^\dagger \vect{\chi}_{s,t}$~\cite{doyle84}.

The \emph{effective $s-t$ resistance} $R_G(\vect{r},s,t)$ of $G$ with respect to the resistances $\vect{r}$ is the potential difference between $s,t$ when we send one unit of electrical flow from $s$ to $t$. That is, $R_G(\vect{r},s,t)=\vect{\phi}(s)-\vect{\phi}(t) = \vect{\chi}_{s,t}^{T}\mat{L}^\dagger \vect{\chi}_{s,t}$, where $\vect{\phi}$ is the vector of vertex potentials induced by the $s-t$ electrical flow of value $1$. We will often denote $R_G(\vect{r},s,t)$ by $R_G(s,t)$ when $\vect{r}$ is clear from the context. It is known that the effective $s-t$ resistance is equal to the energy of the $s-t$ electrical flow of value $1$, that is $R_G(\vect{r},s,t)=\energy_{\vect{r}}(\vect{f},s,t)$.

\paragraph{Graph $r$-Divisions.} 
Let $G=(V,E)$ be a graph. Let $F\subset E$ be a subset of edges. We call the subgraph $G_F$ induced by all edges in $F$ a \emph{region}. For a subgraph $P$ of $G$, any vertex that is incident to vertices not in $P$ is called a \emph{boundary} vertex. The \emph{vertex boundary} of $P$, denoted by $\partial_G(P)$ is the set of boundary vertices belonging to $P$. All other vertices in $P$ will be called \emph{interior vertices} of $P$.

\begin{definition}
	Let $c_1,c_2>0$ be some constant. For any $r\in (1,n)$, a \emph{weak $r$-division} (with respect to $c_1,c_2$) of an $n$-vertex graph $G$ is an edge partition of it into regions $\region=\{P_1,\ldots,P_\ell\}$, where $\ell \leq c_1\cdot \frac{n}{r}$ such that
	\begin{itemize}
		\item Each edge belongs to exactly one region.
		\item Each region $P_i$ contains $r$ vertices.
		\item The total number of all boundary vertices, i.e., $\cup_i\partial_G(P_i)$, is at most $c_2n/\sqrt{r}$. 
	\end{itemize}
\end{definition}

It is known that such an $r$-division (even with the stronger guarantee that each region has $O(\sqrt{r})$ boundary vertices) for planar graphs can be constructed in linear time~\cite{goodrich1995,klein2013,Fed87:shortest}.


\begin{lemma}[\cite{klein2013}] \label{lemma:kms_algorithm}
	Let $c>0$ be some constant. There is an algorithm that takes as input an $n$-vertex planar graph $G$ and for any $r\geq c$, outputs an $r$-division of $G$ in $O(n)$ time.
\end{lemma}

We will need the following property on the boundary vertices of the $r$-division output by the above algorithm (see Section 3.3 in~\cite{klein2013}). 
\begin{lemma}[\cite{klein2013}]~\label{lemma:kms_boundary}
	For an $n$-vertex planar graph $G$, let $\region=\{P_1,\ldots,P_\ell\}$, $\ell = O(n/r)$ be the $r$-division by the algorithm in~Lemma~\ref{lemma:kms_algorithm}. Then it holds that 
	$\sum_{i=1}^\ell|\partial_G(P_i)|=O(n/\sqrt{r}).$
\end{lemma}

\paragraph{Graph Sparsification.} 
Graph Sparsification aims at compressing large graphs into smaller ones while (approximately) preserving some characteristics of the original graph. We present two notions of sparsification. The first requires that the quadratic form of the large and sparsified graph are close. The second requires that all-pairs effective resistances of the corresponding graphs are close. 
\begin{definition}[Spectral Sparsifier] \label{def: specSpar} Let $G=(V,E,\vect{w})$ be a weighted graph and $\varepsilon \in (0,1)$. A $(1 \pm \varepsilon)$-\emph{spectral sparsifier} for $G$ is a subgraph $H=(V,E_H,\vect{w}_H)$ such that for all
	$ \vect{x} \in \mathbb{R}^{n},~(1-\varepsilon)\vect{x}^{T}\mat{L}\vect{x} \leq \vect{x}^{T}\widetilde{\mat{L}}\vect{x} \leq (1+\varepsilon)\vect{x}^{T}\mat{L}\vect{x},$
	where $\mat{L}$ and $\widetilde{\mat{L}}$ are the Laplacians of $G$ and $H$, respectively. 
\end{definition}
\begin{definition}[Resistance Sparsifier] Let $G=(V,E,\vect{w})$ be a weighted graph and $\varepsilon \in (0,1)$. A $(1 \pm \varepsilon)$-\emph{resistance sparsifier} for $G$ is a subgraph $H=(V,E_H,\vect{w}_H)$ such that for all
	$
	u,v \in V,~(1-\varepsilon)R_H(u,v) \leq R_G(u,v) \leq (1+\varepsilon) R_H(u,v), 
	$
	where $R_G(u,v)$ and $R_H(u,v)$ denote the effective $u-v$ resistance in $G$ and $H$, respectively. 
\end{definition}
We remark that Definition \ref{def: specSpar} implies approximations for the pseudoinverse Laplacians, that is 
\[
\forall \vect{x} \in \mathbb{R}^{n} \quad \frac{1}{(1+\varepsilon)}\vect{x}^{T}\mat{L}^\dagger\vect{x} \leq \vect{x}^{T}\widetilde{\mat{L}}^\dagger\vect{x} \leq \frac{1}{(1-\varepsilon)}\vect{x}^{T}\mat{L}^\dagger\vect{x},
\]
Since by definition, the effective resistance between any two nodes $u$ and $v$ is the quadratic form defined by the pseudo-inverse of the Laplacian computed at the vector $\vect{\chi}_{u,v}$, it follows that the effective resistances between any two nodes in $G$ and $H$ are the same up to a $(1 \pm \varepsilon)$ factor. By our definitions for resistance and spectral sparsifiers, we have the following fact.

\begin{fact}
	Let $\varepsilon\in (0,1)$ and let $G$ be a graph. Then every $(1\pm\varepsilon)$-spectral sparsifier of $G$ is a $(1\pm \varepsilon)$-resistance sparsifier of $G$. 
\end{fact}

The following lemma says that given a graph, by decomposing the graph into several  pieces, and computing a good sparsifier for each piece, then one can obtain a good sparsifier for the original graph which is the union of the sparsifiers for all pieces. 

\begin{lemma}[Decomposability] \label{lemm: decomposability} Let $G=(V,E,\vect{w})$ be a graph whose set of edges is partitioned into $E_1,\ldots,E_\ell$. Let $H_i$ be a $(1\pm \varepsilon)$-spectral sparsifier of $G_i=(V,E_i)$, where $i=1,\ldots,\ell$. Then $H=\bigcup_{i=1}^{\ell} H_i$ is a $(1\pm \varepsilon)$-spectral sparsifier of $G$. 
\end{lemma}
\begin{proof}
	For every $i=1,\ldots,\ell$, since $H_i$ is a spectral sparsifier of $G_i$, it follows that 
	\[
	\forall \vect{x} \in \mathbb{R}^{n} \quad(1-\varepsilon)\vect{x}^{T}\mat{L}_i\vect{x} \leq \vect{x}^{T}\widetilde{\mat{L}}_i\vect{x} \leq (1+\varepsilon)\vect{x}^{T}\mat{L}_i\vect{x},
	\]
	where $\mat{L}_i$ and $\widetilde{\mat{L}}_i$ are the Laplacians of $G_i$ and $H_i$, respectively. Summing over these $\ell$ inequalities yields
	\[
	\forall \vect{x} \in \mathbb{R}^{n} \quad(1-\varepsilon)\vect{x}^{T} \sum_{i=1}^{\ell} \mat{L}_i\vect{x} \leq \vect{x}^{T} \sum_{i=1}^{\ell} \widetilde{\mat{L}}_i\vect{x} \leq (1+\varepsilon)\vect{x}^{T} \sum_{i=1}^{\ell} \mat{L}_i\vect{x},
	\]
	where $\mat{L} = \sum_{i=1}^{\ell} \mat{L}_i$ and $\widetilde{\mat{L}} = \sum_{i=1}^{\ell} \widetilde{\mat{L}}_i$. Note that $\mat{L}$ is the Laplacian of $G$ and $\widetilde{\mat{L}}$ is the Laplacian of $H$. This by definition implies that $H$ is a $(1\pm \varepsilon)$-spectral sparsifier of $G$. 
\end{proof}

\section{Dynamic Algorithms for Electrical Flow in Minor-Free Graphs} \label{sec: fullyPlanarEdge}
In order to present our dynamic algorithm for electrical flows, we first introduce the notion of \emph{approximate Schur Complement}.

\subsection{Schur Complement as Vertex Resistance Sparsifier} \label{sec: SchurComplement}
In the previous section we introduced graph sparsification for reducing the number of edges. For our application, it will be useful to define sparsifiers that apart from reducing the number of edges, they also reduce the number of vertices. More precisely, given a weighted graph $G=(V,E,\vect{w})$ with terminal set $K \subset V$,  we are looking for a graph $H=(V_H,E_H,\vect{w}_H)$ with $K \subseteq V_H$ and as few vertices and edges as possible while preserving some important feature among terminals vertices. Graph $H$ is usually referred to as a \emph{vertex sparsifier} of $G$.

\paragraph{Exact Schur Complement.} 
We first review a folklore result~\cite{MillerP13} on constructing vertex sparsifiers that preserve effective resistances among terminal pairs. For sake of simplicity, we first work with Laplcians of graphs. For a given connected graph $G$ as above, let $N = V \setminus K$ be the set of non-terminal vertices in $G$. The partition of $V$ into $N$ and $K$ naturally induces the following partition of the Laplacian $\mat{L}$ of $G$ into blocks:
\[ \mat{L} = 
\begin{bmatrix}
\mat{L}_{N} & \mat{L}_{M} \\
\mat{L}^{T}_{M} & \mat{L}_K \\
\end{bmatrix}
\]
We remark that since $G$ is connected and $N$ and $K$ are non-empty, $\mat{L}_N$ is invertible. We next define the Schur complement of $\mat{L}$, which can be viewed as an equivalent to $\mat{L}$ only on the terminal vertices. 
\begin{definition}[Schur Complement] \label{def: Schur} The \emph{Schur complement} of a graph Laplacian $\mat{L}$ with respect to a terminal set $K$ is
	$
	\mat{L}^{K}_{S} := \mat{L}_{K} - \mat{L}^{T}_{M}\mat{L}^{-1}_{N}\mat{L}_{M}.
	$ 
\end{definition} 
It is known that the matrix $\mat{L}^{K}_{S}$ is a Laplacian matrix for some graph $G'$~\cite{kyng2016}. We can think of Schur Complement as performing Gaussian elimination on the non-terminals $V \setminus K$. This process recursively eliminates a vertex $v \in V \setminus K$ by deleting $v$ and adding a clique with appropriate edge weights on the neighbors of $v$ in the current graph~(see, e.g.~\cite{kyng2016}). The following lemma shows that the quadratic form of the pseudo-inverse of the Laplacian $\mat{L}$ will be preserved by taking the quadratic form of the pseudo-inverse of its Schur Complement, for vectors supported on the terminals. 

\begin{lemma} \label{lemm: SchurComp} Let $\vect{d}$ be a vector of a graph $G$ whose vertices are partitioned into terminals $K$, and non-terminals $N$ and only terminals have non-zero entries in $\vect{d}$. Let $\vect{d}_K$ be the restriction of $\vect{d}$ on the terminals and let $\mat{L}^{K}_S$ be the Schur complement of the Laplacian $\mat{L}$ of $G$ with respect to $K$. Then 
	$
	\vect{d}^{T} \mat{L}^{\dagger} \vect{d} =  \vect{d}_K^{T} (\mat{L}_{S}^{K})^{\dagger} \vect{d}_K.
	$
\end{lemma}
\begin{proof}
	Let $\vect{\phi}$ be the vector of vertex potentials obtained by solving $\mat{L} \vect{\phi} = \vect{d}$. Thus $\vect{\phi}=\mat{L}^{\dagger}\vect{d}$. The vertex partitioning into $N$ and $K$ induces the partitioning of $\vect{\phi}$ into $\vect{\phi}_N$ and $\vect{\phi}_K$, which in turn gives:
	\[  \begin{bmatrix}
	\mat{L}_{N} & \mat{L}_{M} \\
	\mat{L}^{T}_{M} & \mat{L}_K \\
	\end{bmatrix}
	\begin{bmatrix}
	\vect{\phi}_N \\
	\vect{\phi}_K \\
	\end{bmatrix} = 
	\begin{bmatrix}
	\vect{0}\\
	\vect{d}_K \\
	\end{bmatrix}
	\]
	The first row block of $\mat{L}$ gives:
	\begin{align*}
	\mat{L}_{N} \vect{\phi}_N + \mat{L}_{M} \vect{\phi}_K & = \vect{0} \\
	\vect{\phi}_N &= - \mat{L}_{N}^{-1}\mat{L}_{M}\vect{\phi_K}.
	\end{align*}
	Considering the second row block of $\mat{L}$ and substituting $\vect{\phi}_N$ with $- \mat{L}_{N}^{-1}\mat{L}_{M}\vect{\phi_K}$ gives:
	\begin{align*}
	-\mat{L}_{M}^{T}\mat{L}_{N}^{-1}\mat{L}_{M}\vect{\phi_K} + \mat{L}_K \vect{\phi}_K & = \vect{d}_K \\
	(\mat{L}_{K}-\mat{L}^{T}_{M}\mat{L}^{-1}_{N}\mat{L}_{M})\vect{\phi}_K & = \vect{d}_K.
	\end{align*}
	By Definition \ref{def: Schur}, $\mat{L}^{K}_{S} := \mat{L}_{K} - \mat{L}^{T}_{M}\mat{L}^{-1}_{N}\mat{L}_{M}$, and thus it follows that $\mat{L}^{K}_{S} \vect{\phi}_K = \vect{d}_K$ or $\vect{\phi}_K = (\mat{L}_{S}^{K})^{\dagger}\vect{d}_K$. Substituting the latter into $\vect{d}^{T} \mat{L}^{\dagger} \vect{d}$ and using the fact that $\vect{d}$ is non-zero only for the terminals gives:
	\[\vect{d}^{T} \mat{L}^{\dagger} \vect{d} = \vect{d}^T\vect{\phi} =  \vect{d}_K^{T} \vect{\phi}_K = \vect{d}_K^{T}(\mat{L}_{S}^{K})^{\dagger}\vect{d}_K. \]
\end{proof}

Using interchangeability between graphs and their Laplacians, we can interpret the above result in terms of graphs as well. We first present the following notion of sparsification. 
\begin{definition}[Vertex Resistance Sparsifier] Let $G=(V,E,\vect{w})$ be a weighted graph with $K \subset V$ and $\alpha \geq 1$. An $\alpha$-{vertex resistance sparsifier} of $G$ with respect to $K$ is a graph $H=(K,E_H,\vect{w}_H)$ such that for all
	$
	s,t \in K,~R_H(s,t) \leq R_G(s,t) \leq \alpha \cdot R_H(s,t).
	$
\end{definition}

The lemma below relates the Schur Complement and resistance sparsifiers. 
\begin{lemma}~\label{lemma:exact_schur}
	Let $G=(V,E,\vect{w})$ be a weighted graph with $K \subset V$, Laplacian matrix $\mat{L}$ and Schur Complement $\mat{L}^{K}_S$ (with respect to the terminal set $K$). Then the graph $H=(K,E_H,\vect{w}_H)$ associated with the Laplacian $\mat{L}^{K}_S$ is a $1$-vertex resistance sparsifier of $G$ with respect to $K$.
\end{lemma}
\begin{proof}
	Fix some terminal pair $(s,t)$ and consider the vectors $\vect{\chi}_{s,t}$ and $\vect{\chi}'_{s,t}$ of dimension $n$ and $k$, respectively. Lemma \ref{lemm: SchurComp}, the definition of effective resistance and the fact that $\vect{\chi}_{s,t}$ and $\vect{\chi}'_{s,t}$ are valid vectors for $\mat{L}$ and $\mat{L}^{K}_S$ give: $R_G(s,t) = \vect{\chi}_{s,t}^{T}\mat{L}^{\dagger}\vect{\chi}_{s,t} = \vect{\chi'}_{s,t}^{T}\mat{L}_{S}^{K\dagger}\vect{\chi}'_{s,t} =  R_H(s,t).$
\end{proof}

\paragraph{Approximate Schur Complement.} 
Durfee et al.~\cite{DurfeeKPRS16} recently gave a nearly-linear time construction of \emph{approximate Schur Complement} that works even for \emph{general} $k$-terminal graphs. 

\begin{lemma}[\cite{DurfeeKPRS16}] \label{lemm: VertexSpars}
	Fix $\varepsilon \in (0,1/2)$ and $\delta \in (0,1)$. Let $G=(V,E,\vect{w})$ be a weighted graph with $n$ vertices, $m$ edges. Let $K \subset V$ with $|K|=k$. Let $\mat{L}$ be the Laplacian of $G$ and $\mat{L}_S^{K}$ be the corresponding Schur complement with respect to $K$. Then there is an algorithm \textsc{ApproxSchur}$(G,K, \varepsilon, \delta)$ that returns a Laplacian matrix $\widetilde{\mat{L}}^{K}_{S}$ with associated graph $\widetilde{H}$ on the terminals $K$ such that the following statements hold with probability at least $1-\delta$: 
	\begin{enumerate}
		\item The graph $\widetilde{H}$ has $O(k\varepsilon^{-2}\log(n/\delta))$ edges.
		\item $\mat{L}_S^{K}$ and $\widetilde{\mat{L}}^{K}_{S}$ are spectrally close, that is
		\[
		\forall \vect{x} \in \mathbb{R}^{k} \quad (1-\varepsilon)\vect{x}^{T}\mat{L}^{K}_{S}\vect{x} \leq \vect{x}^{T}\widetilde{\mat{L}}^{K}_{S}\vect{x} \leq (1+\varepsilon)\vect{x}^{T}\mat{L}^{K}_{S}\vect{x}.
		\]
	\end{enumerate}
	The total running time for producing $\widetilde{H}$ is $\tilde{O}((n+m)\varepsilon^{-2} \log^{4} (n/\delta))$.
\end{lemma}
In the following, we call the Laplacian $\widetilde{\mat{L}}^{K}_{S}$ (or equivalently, the graph $\widetilde{H}$) satisfying the above two conditions an \emph{approximate Schur Complement} of $G$ with respect to $K$. Note that by definition, the graph $\widetilde{H}$ is a $(1\pm\varepsilon)$-spectral sparsifier of the graph $H$ that is associated with graph $\mat{L}^{K}_{S}$, which in turn is a $1$-vertex resistance sparsifier of $G$ with respect to $K$. Therefore, $\widetilde{H}$ is a $(1\pm O(\varepsilon))$-vertex resistance sparsifier of $G$ with respect to $K$ (see Section \ref{sec:preliminaries}).

\subsection{A Fully Dynamic Algorithm in Minor-Free Graphs: Proof of Theorem~\ref{thm:edge_worst}}
We now present a fully dynamic algorithm for maintaining the energy of electrical flows up to a $(1+\varepsilon)$ factor in minor-free graphs and prove Theorem~\ref{thm:edge_worst}. We start with the special case of planar graphs and give an algorithm with \emph{amortized} update time guarantee. 
\subsubsection{Planar Graphs with Amortized Update Time Guarantee}

\paragraph{Data Structure.} In our dynamic algorithm, we will maintain an $r$-division $\region=\{P_1,\cdots,P_\ell\}$ of $G$ with $\ell=O(n/r)$ and for each region $P_i$, we compute a graph $\widetilde{H}_i$ by invoking the algorithm \textsc{ApproxSchur} in Lemma~\ref{lemm: VertexSpars} with parameters $P_i$, $K=\partial_G(P_i)$, $\varepsilon=\frac{\varepsilon}{6}$ and $\delta=1/n^3$. 

Let $\Ds(G)$ denote such a data structure for $G$, and let $T_{\Ds(G)}$ denote the time to compute $\Ds(G)$. Note that by Lemma~\ref{lemma:kms_algorithm} and~\ref{lemm: VertexSpars}, $T_{\Ds(G)}=\tilde{O}(n+\frac{n}{r}\cdot r\varepsilon^{-2})=\tilde{O}(n\varepsilon^{-2})$. Furthermore, note that there are at most $O(n/r)$ regions, and for each such a region $P_i$, the corresponding graph $\widetilde{H}_i$ is \emph{not} an approximate Schur Complement of $P_i$ with respect to its boundary $\partial_G(P_i)$ with probability at most $1/n^3$. Therefore, by the union bound, with probability at least $1-n\cdot \frac{1}{n^3}=1-\frac{1}{n^2}$, for any $i\leq \ell$, the graph $\widetilde{H}_i$ is an approximate Schur Complement of $P_i$ with respect to $\partial_G(P_i)$, and thus a $(1 \pm \frac{\varepsilon}{6})$-spectral sparsifier of $H_i$, where $H_i$ denotes the exact Schur complement of $P_i$ with respect to $\partial_G(P_i)$. In the following, we will condition on this event. This data structure $\Ds(G)$ will be recomputed every $T_{\divi}:=\Theta(n/r)$ operations. 

\paragraph{Handling Edge Insertions/Deletions.} We now describe the \textsc{Insert} operation. Whenever we compute an approximate Schur Complement, we assume that the procedure \textsc{ApproxSchur} from Lemma \ref{lemm: VertexSpars} is invoked on the corresponding region and its boundary vertex set, with $\varepsilon=\frac{\varepsilon}{6}$ and error probability $\delta=1/n^3$. 
Let us consider inserting an edge $e=(x,y)$.
\begin{itemize}
	\item If both $x,y$ belong to the same region, say $P_i$, then we add the edge $e$ to $P_i$, and recompute an approximate Schur Complement $\widetilde{H}_i$ of the region $P_i$ (with respect to its boundary vertex set) from scratch. 
	\item If $x$ and $y$ do not belong to the same region, we do the following.
	\begin{itemize}
		\item If $x$ is an interior vertex of some region $P_x$, then adding an edge $(x,y)$ will make $x$ a boundary vertex. We then recompute an approximate Schur Complement $\widetilde{H}_x$ of $P_x$.
		\item If $y$ is an interior vertex of some region, then we handle it in the same way as we did for the interior vertex $x$.
		\item We treat the edge $(x,y)$ as a new region containing only this edge.
	\end{itemize}
\end{itemize}
Observe that for each insertion, the number of vertices in any region is always at most $r$, and we perform only a constant number of calls to \textsc{ApproxSchur}, Lemma \ref{lemm: VertexSpars} implies that the time to handle an edge insertion is $\tilde{O}(r\varepsilon^{-2})$. Furthermore, since each edge insertion may increase by a constant the number of boundary nodes and the total number of regions. 

We now describe the \textsc{Delete} operation. If we delete some edge $e=(x,y)$, let $P_i$ be the region such that both $x,y \in P_i$. We remove the edge from $P_i$, and then recompute an approximate Schur Complement $\widetilde{H}_i$ of $P_i$ with respect to its boundary. By Lemma~\ref{lemm: VertexSpars}, the cost of this resparsification step is bounded by $\tilde{O}(r\varepsilon^{-2})$. 

Since we recompute the data structure every $\Theta(n/r)$ operations, the amortized update time is
$\tilde{O}\left(\frac{n\varepsilon^{-2}}{n/r}+ r\varepsilon^{-2}\right)=\tilde{O}(r\varepsilon^{-2})$.


%
%
\paragraph{Handing Queries.} In order to return a $(1+\varepsilon)$-approximation of the energy of $s-t$ electrical flow for an \textsc{ElectricalFlow}($s,t$) query, it suffices to return a $(1-\frac{\varepsilon}{2})$-approximation of the effective $s-t$ resistance, for which we first need to review the static algorithm for computing effective resistance. The following result is due to Durfee et al.~\cite{DurfeeKPRS16} (which builds and/or improves upon~\cite{ChristianoKMST11,KoutisMP14,SpielmanT04}). 

\begin{theorem}[\cite{DurfeeKPRS16}]~\label{thm:effective_resistance}
	Fix $\varepsilon \in (0,1/2)$ and let $G = (V,E,\vect{w})$ be a weighted graph with $n$ vertices and $m$ edges. There is an algorithm \textsc{EffectiveResistance} that computes a value $\psi$ such that 
	$(1-\varepsilon)R_G(s,t) \leq \psi \leq (1+\varepsilon)R_G(s,t),$
	in time $\tilde{O}(m+\frac{n}{\varepsilon^2})$ with high probability.
\end{theorem}

To answer the query \textsc{ElectricalFlow}$(s,t)$, we will form a smaller auxiliary graph that is the union of the regions containing $s,t$ and the approximate Schur Complements of the remaining regions with respect to their boundaries, and output the approximate effective $s-t$ resistance of the smaller graph. More precisely, let $P_s$ and $P_t$ be two regions that contain $s$ and $t$, respectively. Let $\mathcal{J}$ denote the index set of all the remaining regions, i.e., $\mathcal{J}=\{i:P_i\in\region\setminus \{P_s,P_t\}\}$. For each region $P_i$ such that $i\in\mathcal{J}$, as before, let $\widetilde{H}_i$ be the approximate Schur Complement of $P_i$ that we have maintained. Now we form an auxiliary graph $H$ by taking the union over the regions $P_s$ and $P_t$ and all the approximate Schur Complements of the remaining regions, i.e., $H = P_s \cup P_t \cup \bigcup_{i\in \mathcal{J}} \widetilde{H}_i$. 
We then run the algorithm \textsc{EffectiveResistance} on $H$ with $\varepsilon=\frac{\varepsilon}{6}$ to obtain an estimator $\psi$ and return $c_H(s,t):=(1-\frac{\varepsilon}{6})\psi$. Next we show that the returned value is a good approximation to the actual effective resistance.

\begin{lemma}~\label{lem:edge_query}
	Fix $\varepsilon \in (0,1)$. Let $G=(V,E,\vect{w})$ be some current graph and $s,t \in V$. Further, let $H = P_s \cup P_t \cup \bigcup_{i\in \mathcal{J}} \widetilde{H}_i$ be defined as above and let $c_H(s,t)$ be the value returned as above by invoking \textsc{EffectiveResistance} on $H$. Then, with high probability, we get
	\[
	(1-\frac{\varepsilon}{2})R_G(s,t) \leq c_H(s,t) \leq (1-\frac{\varepsilon}{2}) R_G(s,t).
	\]
\end{lemma}
\begin{proof}
	For the sake of analysis, we divide the sequence of updates into intervals each consisting of $T_\divi=\Theta(n/r)$ operations. Let $I$ be the interval in which the query is made. Let $G^{(0)}$ denote the graph at the beginning of $I$. We compute the data structure $\Ds(G^{(0)})$ of $G^{(0)}$, which contains an $r$-division $\region^{(0)}$ and the corresponding approximate Schur Complements $\widetilde{H}_i^{(0)}$. As mentioned before, with probability at least $1-\frac{1}{n^2}$, each of the graphs $\widetilde{H}_i^{(0)}$ will be a $(1\pm \frac{\varepsilon}{6})$-spectral sparsifier of the exact Schur Complement ${H}_i^{(0)}$of the corresponding region with respect to its boundary vertex set. 
	
	Let $G$ be the current graph when the query is made, which is formed from $G^{(0)}$ after some updates in $I$. Let $\region=\{P_i\}_i, \widetilde{H}_i, 1\leq i\leq O(n/r)$ be the $r$-division and the approximate Schur Complements in the current data structure, respectively. Let $H_i$ denote the exact Schur Complement of the region $P_i$ with respect to its boundary vertex set. Since the total number of updates in $I$ is $\Theta(n/r)$, and each update only involves a constant number of invocations of \textsc{ApproxSchur} with error probability $1/n^3$ that recomputes the approximate Schur Complements of some regions, we have that with probability at least $1-O(n/r)\cdot \frac{1}{n^3}\geq 1-\frac{1}{n^2}$, these recomputed approximate Schur Complements are $(1\pm \frac{\varepsilon}{6})$-spectral sparsifiers of the corresponding exact Schur Complements. Therefore, for the current graph $G$ and its data structure, with probability $1-2\cdot \frac{1}{n^2}=1-\frac{2}{n^2}$, for all $i$, the graph $\widetilde{H}_i$ is a $(1\pm \frac{\varepsilon}{6})$-spectral sparsifier of $H_i$. In the following, we will condition on this event.
	
	Recall that $P_s$ and $P_t$ are two regions that contain $s$ and $t$, respectively. Consider the graph $G' = P_s \cup P_t \cup \bigcup_{i\in \mathcal{J}} H_i$. 
	We have the following lemma. 
	\begin{lemma}\label{lemma:exact}
		For any two vertices $u,v\in V(G')$, it holds that $R_{G}(u,v) = R_{G'}(u,v)$. 
	\end{lemma}
	\begin{proof}
		Note that for each $i\in \mathcal{J}$, the exact Schur Complement $H_i$ for $P_i$ with respect to $\partial(P_i)$ is equivalent to recursively performing Gaussian eliminations on all the interior vertices in $P_i$ (see discussions in Section~\ref{sec: SchurComplement}). We consider the following two processes: 
		\begin{enumerate}
			\item for each $i\in \mathcal{J}$, in the subgraph $P_i$, recursively perform Gaussian elimination on all interior vertices $w_1^{(i)},\cdots,w_{k_i}^{(i)} \in V(P_i)\setminus\partial_G(P_i)$, where  the ordering among these vertices is chosen arbitrarily. Note that for each $i$, the resulting graph is the exact Schur Complement $H_i$ with respect to $\partial_G(P_i)$.
			\item in the graph $G$, for each $i\in \mathcal{J}$, recursively perform Gaussian elimination on vertices $w_1^{(i)},\cdots,w_{k_i}^{(i)}$. 
		\end{enumerate}
		
		Note that since for each interior vertex $w$ in $P_i$, all its neighbors are also in $P_i$, performing Gaussian elimination of $w$ in $P_i$ is equivalent to performing Gaussian elimination of $w$ in $G$. Therefore, the resulting graph for $G$ in the second process is the same as $G'=P_s\cup P_t\cup\bigcup_{i\in\mathcal{J}}H_i$, which is the resulting graph in the first process. 
		
		On the other hand, performing Gaussian elimination of all vertices $w_1^{(i)},\cdots,w_{k_i}^{(i)}, i\in \mathcal{J}$ in $G$ results in the exact Schur Complement for $G$ with respect to $V(G')=\cup_{i\in \mathcal{J}} \partial_G(P_i) \cup\{P_s,P_t\}$. That is, $G'$ is an exact Schur Complement of $G$. Then by Lemma~\ref{lemma:exact_schur}, we know that for any $u,v\in V(G')$, $R_G(u,v)=R_{G'}(u,v)$, which finishes the proof of the Lemma.
	\end{proof}
	
	It follows from the above lemma that $R_{G}(s,t) = R_{G'}(s,t)$. We next argue that $H$ is a $(1\pm\frac{\varepsilon}{6})$-resistance sparsifier to $G'$ with high probability. First, note that each of the subgraphs $P_s,P_t$, $H_i$ and $\widetilde{H}_i$, $i \in\mathcal{J}$ can be treated as graphs defined on the same vertex set $V(G')$ with appropriate isolated vertices. 
	Second, since for each $i\in\mathcal{J}$, $\widetilde{H}_i$ is $(1\pm\frac{\varepsilon}{6})$-spectral sparsifier of $H_i$, and $P_s$, $P_t$ are sparsifiers of itself, we know that by Lemma \ref{lemm: decomposability} about the decomposability of sparsifiers, $H$ is a $(1\pm \frac{\varepsilon}{6})$-spectral sparsifier of $G'$. Since every $(1\pm \frac{\varepsilon}{6})$-spectral sparsifier is a $(1\pm \frac{\varepsilon}{6})$-resistance sparsifier, it holds that
	\begin{equation} \label{eqn: 1}
	(1-\frac{\varepsilon}{6})R_H(s,t) \leq R_{G'}(s,t) \leq (1+\frac{\varepsilon}{6}) R_H(s,t).
	\end{equation}
	
	Since by definition we have $c_H(s,t) := (1-\frac{\varepsilon}{6})\psi$, Theorem \ref{thm:effective_resistance} implies that 
	\begin{equation} \label{eqn: 2}
	(1-\frac{\varepsilon}{6})^2R_H(s,t) \leq (1-\frac{\varepsilon}{6}) \psi \leq (1-\frac{\varepsilon}{6}) (1+\frac{\varepsilon}{6})R_H(s,t),
	\end{equation}
	with high probability. Combining~(\ref{eqn: 1}) and~(\ref{eqn: 2}) we get
	\[
	\frac{(1-\frac{\varepsilon}{6})^2}{(1+\frac{\varepsilon}{6})} R_{G'}(s,t) \leq (1-\frac{\varepsilon}{6}) \psi \leq (1+\frac{\varepsilon}{6})R_{G'}(s,t),
	\]
	which in turn along with $R_G(s,t) = R_{G'}(s,t)$ imply that, 
	\[
	(1-\frac{\varepsilon}{2}) R_G(s,t) \leq (1-\frac{\varepsilon}{6}) \psi \leq (1+\frac{\varepsilon}{2})R_G(s,t).
	\]
	Therefore, with high probability, the algorithm outputs a $(1-\frac{\varepsilon}{2})$-approximation to the effective $s-t$ resistance. 
\end{proof}
To bound the query time, we need to bound the size of the $H=P_s \cup P_t \cup \bigcup_{i\in \mathcal{J}} \widetilde{H}_i$. As in the proof of~Lemma~\ref{lem:edge_query}, we let $G^{(0)}$ denote the graph right after the last rebuilding of the data structure. Let $\region^{(0)}$ denote the corresponding $r$-division. By definition, for each $P\in \region^{(0)}$, $|P|\leq r$ and the size of all the boundary vertices is $c_2n/\sqrt{r}$. By Lemma~\ref{lemma:kms_boundary}, we have that $\sum_{P\in\region^{(0)}}|\partial_{G^{(0)}}(P)|\leq O(n/\sqrt{r})$, i.e., the sum of the numbers of boundary vertices over all regions of $G^{(0)}$ is at most $O(n/\sqrt{r})$. 

Note that there will be at most $T_\divi=\Theta(n/r)$ updates between $G^{(0)}$ and $G$, the graph to which the query is performed, and each update can only increase the number of boundary vertices and the total number of regions by a constant. These facts imply that the size of all boundary nodes is $O(n/\sqrt{r})$. Therefore, we have that $|V(H)|\leq O(r+n/\sqrt{r})$, and that the sum of the numbers of boundary vertices of the regions of $G$ is at most $O(n/\sqrt{r})$, i.e., $\sum_{i}|V(\widetilde{H}_i)|\leq O(n/\sqrt{r})$.

On the other hand, by Lemma~\ref{lemm: VertexSpars}, for each $i$, $|E(\widetilde{H}_i)|=O(|V(\widetilde{H}_i)|\cdot \varepsilon^{-2}\log n)$. Thus, 
\begin{eqnarray*}
	|E(H)|&\leq& |E(P_s)|+|E(P_t)|+\sum_{i}|E(\widetilde{H}_i)| \leq O(r) + \sum_{i}|V(\widetilde{H}_i)| \cdot O(\varepsilon^{-2}\log n) \\
	&=&O((r+n/\sqrt{r})\varepsilon^{-2}\log n).
\end{eqnarray*}

By Theorem \ref{thm:effective_resistance}, it follows that the worst-case query time is $\tilde{O}((r+n/\sqrt{r})\varepsilon^{-2})$. Thus, our algorithm has amortized update time $\tilde{O}(r\varepsilon^{-2})$ and worst-case query time $\tilde{O}((r+n/\sqrt{r})\varepsilon^{-2})$.


\subsubsection{From Amortized Update Time to Worst-Case Update Time} \label{app: worstCase}
In the following, we show how to transform the above dynamic algorithm for electrical flows with amortized update time $\tilde{O}(r\varepsilon^{-2})$ into one with the same worst-case update time bound. Our transformation is based on a global rebuilding technique, that has been frequently used for other dynamic graph algorithms.

\paragraph{Data Structure.} The data structure is the same as we maintained before. That is, given the input planar graph $G$, we will compute an $r$-division of $G$ and for each region, an approximate Schur Complement (with respect to the corresponding boundary). Let $\Ds(G)$ denote this data structure.

\paragraph{Updates.} Now we show how to handle updates. Recall that $T_{\divi}=\Theta(n/r)$ denotes the number of operations after which we recompute the data structure in the previous algorithm (with amortized update time guarantee). Let $\Delta=\frac{T_{\divi}}{4}$. For any sequence of updates, we divide it into a number of intervals, each consisting of $T_{\divi}$ updates (except the last interval which may contain less than $T_{\divi}$ updates).   

For each interval $I$ of updates, we further divide it into $3$ subintervals, $I_1,I_2,I_3$ such that $I_1$ contains $2\Delta$ updates and both $I_2$ and $I_3$ contain $\Delta$ updates.

Now consider any interval $I$ with subintervals $I_1,I_2,I_3$. We will maintain two copies of data structures $\Ds_1,\Ds_2$ such that $\Ds_1$ will be used to answer queries, and the second copy $\Ds_2$ is gradually constructed in the subinterval $I_2$ (in each round a $1/\Delta$ fraction of the re-computation is executed) and then it is gradually updated during subinterval $I_3$ (in each round two updates are executed). At the end of $I$, $\Ds_2$ will be ready to serve queries (that occur in the next interval). We next switch the roles of $\Ds_1$ and $\Ds_2$, and continue with the next interval $I'$ and gradually empty the outdated data structure $\Ds_2$ (in the first subinterval of $I'$).  

Let $G_1$ be the graph at the beginning of $I$, and let $G_2$ be the graph at the beginning of $I_2$, i.e., $G_2$ is the graph obtained from $G_1$ after sequentially performing updates in the subinterval $I_1$. 
Note that to construct $G_2$ from $G_1$ with all updates in $I_1$, it takes time $O(\Delta)$, as the number of edges by which they differ is $O(\Delta)$. Recall that it takes time $T_{\Ds(G)}=\tilde{O}(n\varepsilon^{-2})$ to construct the data structure $\Ds(G)$ for a planar graph $G$. If we let $T_\re$ denote the time for first constructing the graph $G_2$ from $G_1$ with all updates in $I_1$ and then computing the data structure $\Ds(G_2)$, then $T_\re=\tilde{O}(\Delta+n\varepsilon^{-2})=\tilde{O}(n\varepsilon^{-2})$. Let $s$ denote the size of the outdated data structure. Note that $s=\tilde{O}(n\varepsilon^{-2})$.

Now let $\Ds_1$ be the data structure for the graph $G_1$ and let $\Ds_2$ be the outdated data structure from the last interval (if there is no interval before $I$, then $\Ds_2$ is empty). Consider the $j$-th update in the interval $I$ for $1\leq j\leq 4\Delta$. We will do the following.

\begin{enumerate}
	\item Update $\Ds_1$ corresponding to the $j$-th update, and also store this update.
	\item If the $j$-th update belongs to $I_1$, i.e., $1\leq j\leq 2\Delta$, then use the next $\frac{s}{2\Delta}$ steps to empty the outdated $\Ds_2$.
	\item If the $j$-th update belongs to $I_2$, i.e., $2\Delta+1\leq j\leq 3\Delta$, then use the next $\frac{T_\re}{\Delta}$ steps for constructing the graph $G_2$, and the data structure $\Ds(G_2)$. Let $\Ds_2=\Ds(G_2)$. 
	\item If the $j$-th update belongs to $I_3$, i.e., $3\Delta+1 \leq j\leq 4\Delta$, then we update $\Ds_2$ with the $(2j-4\Delta-1)$-th and the ($2j-4\Delta)$-th updates of interval $I$.
	\item If $j=4\Delta$, then $\Ds_1=\Ds_2$ and $\Ds_2=\Ds_1$.
\end{enumerate}

Note that during the whole procedure, the data structure $\Ds_1$ contains the $r$-division and the sparsifiers of the current graph $G$, and thus can be used to answer queries. On the other hand, at the end of interval $I$ (and right before we assign $\Ds_2$ to $\Ds_1$), $\Ds_2$ will be \emph{ready} to serve queries, i.e., $\Ds_2=\Ds(G')$ where $G'$ is the graph after all the updates before and in $I$. 

For the running time, note that for each update $j\in I$, we need $\tilde{O}(r\varepsilon^{-2})$ time to update the data structure $\Ds_1$ and $\Ds_2$, due to the recomputation of a constant number of approximate Schur Complements. For the $j$-th update that belongs to $I_1$, we need an additional time $O(\frac{s}{2\Delta})=\tilde{O}(\frac{n\varepsilon^{-2}}{n/r})=\tilde{O}(r\varepsilon^{-2})$ for emptying the outdated structure. For the $j$-th update that belongs to $I_2$, we need an additional time $O(\frac{T_\re}{\Delta})=\tilde{O}(\frac{n\varepsilon^{-2}}{n/r})=\tilde{O}(r\varepsilon^{-2})$ for building $\Ds(G_2)$. Thus, the worst-case update time is 
$\tilde{O}\left( r\varepsilon^{-2}\right)$. 

\paragraph{Query.} To answer a query, we simply invoke the query algorithm for $\Ds_1$, and return the value as before. This will cost time $\tilde{O}((r+n/\sqrt{r})\varepsilon^{-2})$. Note that since we are rebuilding the data structure for each interval of $T_\divi$ updates, we can still guarantee that with high probability, the output is within a $(1+ \varepsilon)$ factor of the energy of $s-t$ electrical flow by the same argument for the proof of Lemma~\ref{lem:edge_query}.

Therefore, we can get worst-case update time $\tilde{O}(r\varepsilon^{-2})$ and worst-case query time $\tilde{O}((r+n/\sqrt{r})\varepsilon^{-2})$. 
\subsubsection{Extension to Minor-Free Graphs} 
In the following, we briefly discuss how one can adapt the previous dynamic algorithms for planar graphs to minor-free graphs. 

The key observation is that since the approximate Schur Complement can be constructed in nearly-linear time for any graph, it suffices for us to efficiently maintain an $r$-division of any minor-free graph, i.e., we need fast algorithms for computing a \emph{separator} of order $\sqrt{n}$ in such graphs. (A separator is a subset $S$ of vertices whose deletion will partition the graph into connected components, each of size at most $\frac{2n}{3}$). Kawarabayashi and Reed~\cite{kawar2010} showed that for any $K_t$-minor-free graph $G$, one can construct in $O(n^{1+\xi})$ time a separator of size $O(\sqrt{n})$ for $G$, for any constant $\xi>0$ and constant $t$. (The $O(\cdot)$ notation for the running time hides huge dependency on $t$.) Applying this separator construction recursively as in Frederickson's algorithm~\cite{Fed87:shortest}, we can maintain an $r$-division of any $K_t$-minor-free $G$ in $\tilde{O}(n^{1+\xi})$ time. Furthermore, by analysis in~\cite{Fed87:shortest}, it holds that the total sum of the sizes of all boundary vertex sets is also bounded by $O(\frac{n}{\sqrt{r}})$ as guaranteed by Lemma~\ref{lemma:kms_boundary} for planar graphs. 

Now we can dynamically maintain the data structure for electrical flows in minor-free graphs almost the same as we did for planar graphs, except that we use the above $\tilde{O}(n^{1+\xi})$ time algorithm to compute the $r$-divisions. Thus, the time to compute the data structure for any minor-free graph is then $\tilde{O}(n^{1+\xi}+n\varepsilon^{-2})$, for arbitrarily small constant $\xi>0$. Then from previous analysis, the worst-case update time is $\tilde{O}(n^{\xi}r\varepsilon^{-2})$ update time, and the worst-case query time is $\tilde{O}((r+n/\sqrt{r})\varepsilon^{-2})$. This completes the proof of Theorem~\ref{thm:edge_worst}.
\subsection{Dynamic Electrical Flows in the Incremental Subgraph Model}\label{sec:subgraph}

In the following, we will give a data structure for maintaining all-pairs electrical flows in the \emph{vertex-update} version of dynamic graphs, also called the \emph{subgraph} model~\cite{FrigioniI00}, for which the updates consist of vertex activations and/or deactivations. Furthermore, we restrict to the \emph{incremental subgraph} model, where update operations consist of activations only. We remark that the decremental subgraph model can be studied similarly while we omit the details here. More specifically, we allow the following operations:
%
\begin{itemize}
	\itemsep0.1em 
	\item \textsc{Activate}$(u)$: Activate the vertex $u$ along with its incident edges in $G$.
	\item \textsc{ElectricalFlow}$(s,t)$: Return the exact (or approximate) energy of the $s-t$ electrical flow in the subgraph induced by all active nodes.  
\end{itemize}
Note that the queries should refer to the subgraph induced by all current active vertices. 

\subsubsection{An Upper Bound}
The performance of our dynamic algorithm is formalized in the following theorem. 

\begin{theorem}~\label{thm:alg_subgraph}
Let $r\geq c$ for some large constant $c>0$. One can maintain a data structure for the incremental subgraph all-pairs electrical flow problem for minor-free graphs that supports an activate operation in $\tilde{O}(r\varepsilon^{-2})$ \emph{amortized} time, and with high probability, outputs an answer $f$ to the \textsc{ElectricalFlow}($s,t$) query in $\tilde{O}((r+\frac{n}{\sqrt{r}})\varepsilon^{-2})$ \emph{worst-case} time such that 
	$$\energy_{G[S]}(s,t)\leq f\leq (1+\varepsilon)\energy_{G[S]}(s,t),$$
	where $G[S]$ is the graph induced by all current active vertices $S$.
\end{theorem}

Now we present the proof of Theorem~\ref{thm:alg_subgraph}. We first describe the data structure and the corresponding algorithm.

\paragraph{Preprocessing.} We compute a weak $r$-division of the input minor-free graph $G$ (treating all nodes as active) and let $\region=\{P_1,\cdots,P_\ell\}$, $\ell=O(n/r)$ denotes the resulting set of regions. Let $S$ be the set of vertices that are active. Initially, $S$ is empty.

\paragraph{Handling Activation Operations.} If vertex $v$ is activated, then we add $v$ to $S$. Let $G[S],P_i[S]$ denote the subgraphs of $G$ and $P_i$ induced by all active vertices in $S$, respectively. If $v$ is an interior vertex, then we recompute the approximate Schur Complement for the region $P_i$ that contains $v$, by invoking the algorithm \textsc{ApproxSchur} in Lemma~\ref{lemm: VertexSpars} with parameters $P_i[S]$, $\partial_{G[S]}(P_i[S]),\varepsilon=\frac{\varepsilon}{6}, \delta=1/n^3$. If $v$ is a boundary vertex, then for each $i$ such that $P_i$ that contains $v$, we recompute the approximate Schur Complement $\widetilde{H}_i[S]$ for the region $P_i[S]$, by invoking \textsc{ApproxSchur} with parameters $P_i[S]$, $\partial_{G[S]}(P_i[S]),\varepsilon=\varepsilon/6, \delta=1/n^3$.

Note that if an interior vertex is activated, we only need to recompute the approximate Schur Complement for a single region, which costs time $T_\asc:=\tilde{O}(r\varepsilon^{-2})$. On the other hand, if some boundary vertex $v$ is activated, then we need to recompute all the approximate Schur Complements for regions (induced by active vertex set $S$) that contain $v$ in the division. Let $b(v)$ denote the number of regions containing $v$. Then the cost of recomputing the sparsifiers for all sub-regions that contain $v$ is at most $b(v)\cdot T_\asc$. 

Let $V_B$ denote the set of all boundary vertices in the $r$-division. Note that the total cost for updating all active vertices is upper bounded by the total time of recomputing approximate Schur Complements after each vertex activation. This is upper bounded by 
\[ T_\asc \cdot (n-|V_B|) + T_\asc\cdot \sum_{v\in V_B} b(v) \leq  T_\asc\cdot O(n+\frac{n}{\sqrt{r}}) =\tilde{O}(nr\varepsilon^{-2}),
\]
where first inequality follows from the corresponding argument in the proof of Theorem~\ref{thm:edge_worst}. 
Note that there are $n$ activation updates, so the amortized update time is $\tilde{O}(r\varepsilon^{-2})$.

\paragraph{Handling Queries.} To answer the query \textsc{ElectricalFlow}$(s,t)$ when both $s,t$ are active, we can compute the answer almost the same as we did in the edge-update model. More precisely, given the current active vertex set $S$, we first form an auxiliary graph $H^{(S)}$ by taking the union over the regions $G_1[S]$ and $G_2[S]$ containing $s$ and $t$ and all the approximate Schur Complements of the remaining regions (induced by active vertices) with respect to the corresponding boundaries, i.e., $H^{(S)}=P_s[S]\cup P_t[S]\cup \bigcup_{i\in \mathcal{J}}\widetilde{H}_i[S]$, where $\mathcal{J}=\{i:P_i\in\region\setminus \{P_s,P_t\}\}$. We then run the algorithm \textsc{EffectiveResistance} in Theorem~\ref{thm:effective_resistance} on the graph $H^{(S)}$ with $\varepsilon=\frac{\varepsilon}{6}$ to obtain an estimator $\psi$ and return $(1-\frac{\varepsilon}{6})\psi$.

Note that since the total number of invocations of \textsc{ApproxSchur} is $n-|V_B|+\sum_{v\in V_B}b(v)=O(n)$, and each invocation has error probability at most $1/n^3$, we can guarantee that with probability at least $1-O(1/n^2)$, all the approximate Schur Complements are $(1\pm\frac{\varepsilon}{6})$-spectral sparsifiers of the corresponding exact Schur Complements. Then by the same argument as in the proof of Lemma~\ref{lem:edge_query}, we can guarantee that the output is within a $(1-\frac{\varepsilon}{2})$ factor of $R_{G[S]}(s,t)$, and thus within a $(1+\varepsilon)$ factor of $\energy_{G[S]}(s,t)$.

The number of edges of the graph $H^{(S)}$ can again be bounded by $O((r+n/\sqrt{r})\varepsilon^{-2}\log n)$, since each region has size $O(r)$, the total number of boundary vertices is $O(n/\sqrt{r})$ at all times in the algorithm, and the number of edges of any approximate Schur Complement $\widetilde{H}_i^{(S)}$ is at most $O(\varepsilon^{-2}\log (n/\delta))$ times the number of its vertices by Lemma~\ref{lemm: VertexSpars}. Therefore, the running time to answer an \textsc{ElectricalFlow}$(s,t)$ query is again upper bounded by $\tilde{O}((r+n/\sqrt{r})\varepsilon^{-2})$ by Theorem~\ref{thm:effective_resistance}. 
%
This completes the proof of Theorem~\ref{thm:alg_subgraph}.
	
%
%
%

\subsubsection{Conditional Lower Bound for Partial Subgraph Electrical Flows}
\label{sec:lower_bound}
We complement the above result by giving a conditional lower bound for the $s-t$ electrical flow problem in the incremental subgraph model for \emph{general} graphs. Our lower bound is based on a reduction to the online boolean matrix-vector (OMv) multiplication problem, which is conjectured to have no truly subcubic time algorithm \cite{henzinger2015unifying}. 
We first introduce some basic definitions.
\begin{definition}
In the \emph{Online Boolean Matrix-Vector Multiplication (OMv)} problem, we are given an integer $n$ and an $n\times n$ Boolean matrix $\mat{M}$. Then at each step $i$ for $1\leq i\leq n$, we are given an $n$-dimensional column vector $\vect{v}_i$, and we should compute $\mat{M}\vect{v}_i$ and output the resulting vector before we proceed to the next round.
\end{definition}

\begin{conjecture}[OMv conjecture~\cite{henzinger2015unifying}]~\label{conjecture}
For any constant $\varepsilon>0$, there is no $O(n^{3-\varepsilon})$-time algorithm that solves OMv with error probability at most $1/3$. 
\end{conjecture}

Our main result in this section is the following conditional lower bound.

\begin{theorem}~\label{thm:subgraph_lower}
	For any $C>0$ and constant $\varepsilon>0$, there is no incremental algorithm that computes a $C$-approximation of the energy of the $s-t$ electrical flows in general graphs that has both $O(n^{1-\varepsilon})$ worst-case update time and $O(n^{2-\varepsilon})$ worst-case query time, unless the OMv conjecture is false. 
\end{theorem}
Note that our algorithm in Theorem~\ref{thm:alg_subgraph} and the lower bound in Theorem~\ref{thm:subgraph_lower} demonstrate potentially a polynomial gap for dynamic algorithms for subgraph electrical flows between minor-free graphs and general graphs.  

To prove Theorem~\ref{thm:subgraph_lower}, we will work on a related problem which is called the $\vect{u}\mat{M}\vect{v}$ problem.
\begin{definition}
	In the $\vect{u}\mat{M}\vect{v}$ problem with parameters $n_1,n_2$, we are given a matrix $M$ of size $n_1\times n_2$ which can be preprocessed. After preprocessing, a vector pair $\vect{u},\vect{v}$ is presented, and our goal is to compute $\vect{u}^T\mat{M}\vect{v}$.
\end{definition}

\begin{theorem}[\cite{henzinger2015unifying}]~\label{thm:umv_hard}
Unless OMv conjecture~\ref{conjecture} is false, there is no algorithm for the $\vect{u}\mat{M}\vect{v}$ problem with parameters $n_1,n_2$ using polynomial preprocessing time and computation time $O(n_1^{1-\varepsilon}n_2+n_1n_2^{1-\varepsilon})$ that has an error probability at most $1/3$, for some constant $\varepsilon$.
\end{theorem}

We now reduce the problem $\vect{u}\mat{M}\vect{v}$ problem with parameters $n_1=n_2=n$ to the $s-t$ electrical flow problem in incremental subgraph model (similar reduction can be built for decremental subgraph model). Our construction is the same as the construction for $st$-subgraph connectivity in~\cite{henzinger2015unifying}.

\begin{proof}[Proof of Theorem~\ref{thm:subgraph_lower}]
Consider any $\vect{u}\mat{M}\vect{v}$ instance with parameters $n_1=n_2=n$. Given $M$, we construct $n$ row vertices $r_1,\cdots,r_n$ and $n$ column vertices $c_1,\cdots,c_n$. We construct a bipartite unweighted graph between row vertices and column vertices such that we add an edge $(r_i,c_j)$ if and only if $M(i,j)=1$. Now we add two special vertices $s,t$. We connect $s$ to every row vertex and connect $t$ to every column vertex.
	
	Now to compute $\vect{u}^T\mat{M}\vect{v}$, it suffices to activate all the row vertices $r_i$ such that $\vect{u}(i)=1$ and all the column vertices $c_j$ such that $\vect{v}(j)=1$, and query if $s-t$ electrical flow energy is at most $C\cdot {m}$, where $m$ is the number of edges the constructed graph. Observe that $\vect{u}^T\mat{M}\vect{v}=1$ if and only if there is a path from $s$ to $t$. Furthermore, if there is no such path, then the answer to the query will be infinity; if there is a path from $s$ to $t$, then by definition, the electrical flow energy is at most $m$. Any $C$-approximation algorithm for the energy of electrical flows will successfully distinguish if there is a path from $s$ to $t$ or not, and thus compute the answer $\vect{u}^T\mat{M}\vect{v}$.
		
	Note that the total number of activation operations is $O(n)$ and the number of queries is $1$. If there exists a $C$-approximation algorithm for electrical flow with $O(n^{1-\varepsilon})$ update worst-case time and $O(n^{2-\varepsilon})$ query worst-case time, then we can also solve $\vect{u}^T\mat{M}\vect{v}$ problem with parameters $n_1=n_2=n$ in $O(n^{2-\varepsilon})$ time, which is a contradiction.
\end{proof}








\section{Dynamic All-Pairs Max Flow in Minor-Free Graphs}\label{sec:max_flow_minor}
In this section, we present a dynamic algorithm for maintaining all-pairs max flow in minor-free graphs and prove Theorem~\ref{thm:maxflow_1}. Besides the $r$-division of such graphs, the main building block in our construction is the usage of vertex cut sparsifier~\cite{Moitra09}, which we define below.  

\paragraph{Terminal Minimum Cuts.} Let $G=(V,E,c)$ be an undirected graph with terminal set $K \subset V$ of
cardinality $k$, where $c: E \rightarrow \mathbb{R}_{\geq 0}$ assigns a non-negative
capacity to each edge. Let $U \subset V$ and $\emptyset \neq S \subset K$. We say that a cut $(U, V \setminus U)$ is
\emph{$S$-separating} if it separates the terminal subset $S$ from its complement $K
\setminus S$, i.e., $U \cap K$ is either $S$ or $K \setminus S$. We will refer to such a cut as a \emph{terminal cut}. The cutset
$\delta(U)$ of a cut $(U, V \setminus U)$ represents the edges that have one
endpoint in $U$ and the other one in $V \setminus U$. The cost
$\capacity_G(\delta(U))$ of a cut $(U, V \setminus U)$ is the sum over all
capacities of the edges belonging to the cutset. We let $\text{mincut}_{G}(S, K
\setminus S)$ denote the minimum cost of any $S$-separating cut of
$G$. 

\paragraph{Vertex Cut Sparsifiers.} A graph $H = (V_H, E_H, c_H)$, $K \subset V_H$ is a \emph{vertex cut
  sparsifier} of $G$ with \emph{quality} $q \geq 1$ if for any $\emptyset\neq S \subset K$,
\[ \mincut_G(S, K \setminus S) \leq \mincut_H(S, K \setminus S) \leq q \cdot
\mincut_G(S, K \setminus S). 
\]

Next we review two different constructions of vertex cut sparsifiers. Under the assumption that the input graph is minor-free, the first construction due to Moitra~\cite{moitra2011vertex} produces a sparsifier $H$ only on the terminals with $O(1)$ quality in time $\tilde{O}(km^2)$. We remark that w.l.o.g. we can assume that $H$ has only $\tilde{O}(k)$ edges. This can be achieved by running an edge cut-sparsifier~\cite{BenczurK96} on top of $H$, thus reducing the number of edges to $\tilde{O}(k/\varepsilon^2)$ while paying only a $(1+\varepsilon)$ multiplicative factor in the quality guarantee of the sparsifier, for any error parameter $\varepsilon$. Since the best-known algorithm to construct a vertex cut sparsifier requires $\Omega(km^{2})$ time, the $\tilde{O}(k^{2})$ term that comes from the construction of an edge-cut sparsifiers is dominated. These guarantees are summarized in the following lemma.

\begin{lemma}~\label{lemm:cutspar1}
Let $G=(V,E,c)$ be a capacitated minor-free graph with $m$ edges, where $K \subset V$ is a subset of terminals with $|K| = k$. Then there is an algorithm \textsc{CutSparsifyI}$(G,K)$ that returns a weighted graph $H$ on the terminals $K$ such that the following statements hold:
\begin{enumerate}
\item The graph $H$ has $\tilde{O}(k)$ edges.
\item The graph $H$ is an $O(1)$-quality vertex cut sparsifier of $G$.
\end{enumerate}
The total time for producing $H$ is $\tilde{O}(km^{2})$. 
\end{lemma}

Using the above vertex sparsification, we can obtain a dynamic algorithm for $O(1)$-approximating all-pairs max flow with a trade-off between the update time and query time as guaranteed in the first part of Theorem~\ref{thm:maxflow_1}.

If we are willing to lose a polylogarithmic approximation factor, then we can obtain a dynamic algorithm for all-pairs max flow with a better update/query trade-off as guaranteed in the second part of Theorem~\ref{thm:maxflow_1}. This is achieved by invoking another vertex cut sparsifier,  due to Peng~\cite{Peng16} and R\"acke et al.~\cite{RackeST14}, who showed how to bring down the construction time to $\tilde{O}(m)$ at the cost of obtaining only poly-logarithmic quality. 

\begin{lemma}~\label{lemm:cutspar2}
	Let $G=(V,E,c)$ be a capacitated graph, where $K \subset V$ is subset of terminals with $|K| = k$. Then there is an algorithm \textsc{CutSparsifyII}$(G,K)$ that returns a weighted tree $H$ with $K \subset V_H$  such that the following statements hold:
	\begin{enumerate}
		\item The graph $H$ has $O(k)$ vertices and $O(k)$ edges.
		\item The graph $H$ is an $O(\log^{4} n)$-quality vertex cut sparsifier of $G$.
	\end{enumerate}
	The total time for producing $H$ is $\tilde{O}(m)$. 
\end{lemma}



\subsection{Proof of Theorem~\ref{thm:maxflow_1}}
In the following, we prove Theorem~\ref{thm:maxflow_1}. We will focus on the first part of the theorem by using the cut sparsifiers as guaranteed in Lemma~\ref{lemm:cutspar1}. The proof of the second part is almost the same, except that we use Lemma~\ref{lemm:cutspar2} instead in our construction and analysis. We omit the details of the proof of the second part.
\paragraph{Data Structure.} As before, we will maintain an $r$-division $\mathcal{P} = \{P_1,\ldots, P_{\ell}\}$ of $G$ with $\ell =O(\frac{n}{r})$ such that $|V(P_i)|\leq r$ and $|\partial_G(P_i)|\leq O(\sqrt{r})$, which can be initially computed in time $\tilde{O}(n^{1+\xi})$ for any $\xi>0$ by applying Kawarabayashi and Reed's separator construction recursively as in Frederickson's algorithm~\cite{kawar2010,Fed87:shortest}. Note that $\sum_{i}|\partial_G(P_i)|\leq O(n/\sqrt{r})$. Next for each region $P_i$, we compute a graph $\widetilde{H}_i$ by invoking the algorithm \textsc{CutSparsifyI} in Lemma~\ref{lemm:cutspar1} with parameters $P_i, K = \partial_G(P_i)$. We let $\Ds(G)$ denote the resulting data structure of $G$. By Lemma~\ref{lemm:cutspar1}, the time to construct such a sparsifier is $\tilde{O}(|\partial_G(P_i)|\cdot |P_i|^{2})=\tilde{O}(r^{2}|\partial_G(P_i)|)$. Thus, the expected time to compute $\Ds(G)$ is \[\tilde{O}(n^{1+\xi} +  \sum_ir^{2}\cdot|\partial_G(P_i)|)=\tilde{O}(n^{1+\xi} + n\cdot r^{1.5})\]
This data structure $\Ds(G)$ will be recomputed every $T_{\divi}:=\Theta(n/r)$ operations to guarantee that the number of regions remains $O(n/r)$.

\paragraph{Handling Edge Insertions/Deletions.} We handle edge insertions and deletions almost the same as before, except that we replace the algorithm \textsc{ApproxSchur} by the algorithm \textsc{CutSparisfyI} when we resparsify the affected regions. Note that at any time, each region has at most $r$ vertices and each update requires only to resparsify a constant number of regions and for each such region $P$, it can be resparsified in $\tilde{O}(b|P|^2)=\tilde{O}(br^2)$ time, where $b:=O(\min\{r,\sqrt{r}+\frac{n}{r} \})$ is an upper bound on the number of boundary vertices of any region. Since we recompute the data structure every $\Theta(n/r)$ operations, the amortized update time is 
\[
\tilde{O}\left(\frac{n^{1+\xi} + n\cdot r^{1.5}}{n/r}+br^2\right)=\tilde{O}(n^\xi r+r^3).
\]

\paragraph{Handling Queries.} To answer the query \textsc{MaxFlow}$(s,t)$, similarly as before, we first form an auxiliary graph $H$ that is the union of the regions $P_s,P_t$ containing $s,t$, respectively, and the vertex cut sparsifiers of the remaining regions with respect to their boundary nodes. More formally, let $\mathcal{J}$ denote the index set of all the remaining regions, i.e., $\mathcal{J}=\{i:P_i\in\region\setminus \{P_s,P_t\}\}$. For each region $P_i$ such that $i\in\mathcal{J}$, as before, let $\widetilde{H}_i$ be the vertex cut sparsifier of $P_i$ that we have maintained. For any subgraph $P$ of $G$, we let $G_{P}$ denote the graph with the same vertex set $V(G)$ as $G$ and edge set $E(P)$. We let $K$ denote the set of all boundary vertices in any $P_i$ as well as vertices $\{s,t\}$. We will treat $K$ as the \emph{terminal set} of the graphs $G$ and $H$. We have the following useful lemma.

\begin{lemma}~\label{lem:union_maxflow}
Fix $\varepsilon \in (0,1)$. Let $G=(V,E,\vect{w})$ be some current graph and $s,t \in V$. Further, let the graph $H = P_s \cup P_t \cup \bigcup_{i\in \mathcal{J}} \widetilde{H}_i$ and the set $K$ be defined as above. Suppose that each $\widetilde{H}_i$ is $q$-quality vertex cut sparsifier of $P_i$. Then for any $S\subset K$, it holds that 
\[\mincut_G(S,K\setminus S)\leq \mincut_{H}(S,K\setminus S)\leq q\cdot \mincut_G(S,K\setminus S) \]
\end{lemma}
\begin{proof}
We write $G=P_s\cup P_t\cup\bigcup_{i\in \mathcal{J}}P_i$. Let $S\subset K$ be any subset of terminal vertices. Consider a minimum $S$-separating cut $(U,V\setminus U)$ in $G$. Note that when restricting to the subgraph $G_{P_i}$, the cut $(U,V\setminus U)$ is also a minimum $S$-separating cut, for any $i\in \mathcal{J}\cup\{s,t\}$. This is true since otherwise for at least one $G_{P_j}$, there exists some other $S$-separating cut $(U',V\setminus U')$ such that the cost of $(U',V\setminus U')$ is smaller than the cut $(U,V\setminus U)$ and only interior vertices in $P_j$ can be contained in $(U\setminus U')\cup (U'\setminus U)$. The latter is true since there is no edge incident to interior vertices in any other region $P_i$ with $i\neq j$. As all edges in different regions are disjoint, this further implies that $(U',V\setminus U')$ is a valid $S$-separating cut in $G$ with smaller cost than $(U,V\setminus U)$, which is a contradiction. 

Now let us consider the minimum $S$-separating cut $(\widetilde{U},\widetilde{V}\setminus \widetilde{U})$ in $H$, where $\widetilde{V}$ denote the vertex set of $H$. Again, for each $i\in \mathcal{J}\cup\{s,t\}$, the cut $(\widetilde{U},\widetilde{V}\setminus \widetilde{U})$ is also a minimum $S$-separating cut in $H_{P_i}$. Since for each $i\in \mathcal{J}$, $\widetilde{H}_i$ is a quality $q$ cut sparsifier of $P_i$, then we have that for each $i\in \mathcal{J}$,
\[\capacity_{G_{P_i}}(\delta(U))\leq \capacity_{H_{\widetilde{H}_i}}(\delta(\widetilde{U}))\leq q\cdot \capacity_{G_{P_i}}(\delta(U))\]
We also have that for $i\in \{s,t\}$, the minimum $S$-separating cut is preserved exactly from $G_{P_i}$ to $H_{P_i}$, that is,
\[\capacity_{G_{P_i}}(\delta(U))= \capacity_{H_{P_i}}(\delta(\widetilde{U})).\]

By the fact that the edges in different regions are disjoint, we have that 
\[
\capacity_G(\delta(U))=\sum_{i\in \mathcal{J}\cup\{s,t\}}\capacity_{G_{P_i}}(\delta(U)), \qquad
\capacity_H(\delta(\widetilde{U}))=\sum_{i\in \mathcal{J}\cup\{s,t\}}\capacity_{H_{\widetilde{H}_i}}(\delta(\widetilde{U})).
\]

Therefore,
\[
\capacity_G(\delta(U))\leq \capacity_H(\delta(\widetilde{U})) \leq q \capacity_G(\delta(U)),
\]
or equivalently,
\[\mincut_G(S,K\setminus S)\leq \mincut_{H}(S,K\setminus S)\leq q\cdot \mincut_G(S,K\setminus S). \]
\end{proof}

Note that the above lemma implies that the value of the $s-t$ min cut in $H$ is at least the value of $s-t$ min cut in $G$ and at most a $q$ factor of the value the $s-t$ min cut in $G$. 

Now we observe that at any time, the total number of regions is at most $O(\frac{n}{r})$, each region has size at most $r$, and the sum of the number of boundary vertices over all regions is at most $O(n/\sqrt{r})$. Therefore, the number of vertices and the number of edges in $H$ are both at most $O(r)+\tilde{O}(\sum_i |\partial_G(P_i)|)=\tilde{O}(r+n/\sqrt{r})$.

Finally, we invoke the $(1+\varepsilon)$-approximation algorithm for $s-t$ max flow~\cite{Peng16} on top of $H$ and output $f_H(s,t)$, by setting $\varepsilon$ to be some arbitrarily small constant. This algorithm runs in time $\tilde{O}(|V(H)|+|E(H)|)=\tilde{O}(r+n/\sqrt{r})$, and the resulting output is a $(1+\varepsilon)$-approximation of the value of the $s-t$ max flow of $H$. By Lemma~\ref{lem:union_maxflow} and the fact that $q=O(1)$, $f_H(s,t)$ is an $O(1)$-approximation of the value of the $s,t$-maximum flow of $G$. 

Therefore, the above algorithm has amortized update time $\tilde{O}(n^\xi r+r^3)$ and worst-case query time $\tilde{O}(r+n/\sqrt{r})$. Using the global rebuilding technique, described in Section~\ref{app: worstCase}, we can guarantee asymptotically the same worst-case update time as before. This completes the proof of the first part of Theorem~\ref{thm:maxflow_1}.

\section{Dynamic All-Pairs Shortest Path in Minor-Free graphs}\label{sec:distance_minor}
In this section we show that by employing our algorithmic framework for dynamically maintaining the electrical flows in minor-free graphs, we can obtain a fully-dynamic algorithm for maintaining an approximation to the all-pairs shortest path problem in minor-free graphs while achieving sublinear update and query time. We will prove Theorem~\ref{thm:shortest_path}. We will maintain an $r$-division of the minor-free graph, and use as a building block \emph{vertex distance sparsifiers} of good quality that can be constructed efficiently.



\paragraph{Vertex Distance Sparsifier.} In the following we show an analogue of Lemma~\ref{lemm: VertexSpars} for distances, which applies to minor-free graphs.
\begin{lemma} \label{lemm: VertexDistanceSpars}
Fix $q \geq 1$. Let $G=(V,E,\vect{w})$ be a weighted minor-free graph with $n$ vertices. Let $K \subset V$ be a set of terminal vertices with $|K|=k$. Then there is an algorithm \textsc{DistanceSparsify}$(G,K,q)$ that returns a weighted graph $\widetilde{H}$ on the terminals $K$ such that the following statements hold:
\begin{enumerate}
\item The graph $\widetilde{H}$ has $O(qk^{1+1/q})$ edges.
\item Shortest path lengths among terminal pairs in $G$ and $\widetilde{H}$ are close, that is
\[
	\forall s,t \in K \quad d_G(s,t) \leq d_{\widetilde{H}}(s,t) \leq (2q-1) \cdot d_G(s,t).
\]
\end{enumerate}
The total \emph{expected} time for producing $\widetilde{H}$ is ${O}(n^{3/2} + k^2)$.
\end{lemma}
\begin{proof}[Proof Sketch]
Our algorithm consists of two subroutines. First, for the input graph $G$, we compute a complete graph $H$ on the terminals $K$ that preserves exactly all-pair shortest path lengths among terminals. Second, we compute a \emph{$(2q-1)$-spanner} $\widetilde{H}$ of $H$, which is a subgraph of $H$ with $O(q\cdot k^{1+1/q})$ edges and preserves all the pairwise distances of $H$ within a factor of $2q-1$, for any $q\geq 1$ (see~\cite{awerbuch1985}). Note that for any $s,t \in K$, $d_G(s,t) = d_H(s,t)$, and that $d_H(s,t) \leq d_{\widetilde{H}}(s,t) \leq (2q-1) \cdot d_H(s,t)$. This proves the first and second statements of the lemma. 

Now we bound the running time of the algorithm. To construct $H$ from $G$, we repeatedly pick a non-terminal $v \in V \setminus K$ and replace $v$ by a clique with appropriate edge weights among its neighbours. If this operation produces parallel edges, only the edge with the smallest weight is kept. It is not hard to see that all-pairs terminal distances are not affected by this operation and the final graph obtained after $n-k$ operations is a complete graph only on the terminals. For a more detailed description of this procedure, we refer the reader to~\cite{HenzingerKRS97}. Now, since minor-free graphs admit balanced separators of size $O(\sqrt{n})$ (and those can be found in $\tilde{O}(n^{1+\xi})$ time, for any $\xi > 0$~\cite{kawar2010}), we can use the nested-dissection algorithm due to Lipton, Rose and Tarjan~\cite{LiptonRT79} to finish all these $n-k$ operations in $O(n^{3/2})$ time. To construct $\tilde{H}$ from $H$, we apply the (expected) linear time algorithm for constructing a $(2q-1)$-spanner by Baswana and Sen~\cite{BaswanaS07}, on the graph $H$, which takes time $O(k^2)$ in expectation, since $H$ has $O(k^2)$ edges. Therefore, the total expected time for constructing $\widetilde{H}$ is $O(n^{3/2} + k^2)$.
\end{proof}

\subsection{Proof of Theorem~\ref{thm:shortest_path}}
Now we are ready to prove Theorem~\ref{thm:shortest_path}. We first describe the data structure and the corresponding dynamic algorithm.
\paragraph{Data Structure.} We will maintain an $r$-division $\mathcal{P} = \{P_1,\ldots, P_{\ell}\}$ of $G$ with $\ell = O(n/r)$. 
As before, we initialize the data structure by combining the $\tilde{O}(n^{1+\xi})$ time algorithm for computing a separator of size $O(\sqrt{n})$ for minor-free graphs~\cite{kawar2010} and Frederickson's algorithm~\cite{Fed87:shortest}, for any $\xi>0$, which outputs an $r$-division $\{P_1,\cdots, P_\ell\}$ with at most $O(n/r)$ regions such that each region has at most $r$ vertices and at most $O(\sqrt{r})$ boundary vertices. This also implies that $\sum_{i}|\partial_G(P_i)|\leq O(n/\sqrt{r})$, and that $\sum_{i}|\partial_G(P_i)|^2 \leq r \cdot \sum_{i}|\partial_G(P_i)|= O(n\sqrt{r})$ after initialization. The running time for finding such an $r$-division is still bounded by $\tilde{O}(n^{1+\xi})$.

Next for each region $P_i$, we compute a graph $\widetilde{H}_i$ by invoking the algorithm \textsc{DistanceSparsify} in Lemma~\ref{lemm: VertexDistanceSpars} with parameters $P_i, K = \partial_G(P_i)$ and $q$, where $q \geq 1$. We let $\Ds(G)$ denote the resulting data structure of $G$. By Lemma~\ref{lemm: VertexDistanceSpars}, the expected time to construct such a sparsifier is $\tilde{O}(\sum_{i=1}^{\ell} (r^{3/2} + k_i^{2}))$, where $k_i = |\partial_G(P_i)|$. Thus, the expected time to compute $\Ds(G)$ is 
\[
\tilde{O}(n^{1+\xi} + \sum_{i=1}^{\ell} (r^{3/2} + k_i^{2}))
=
\tilde{O}(n^{1+\xi}+\frac{n}{r}\cdot r^{3/2}+n\sqrt{r}) 
=
\tilde{O}(n^{1+\xi}+n\sqrt{r}).
\]


This data structure $\Ds(G)$ will be recomputed every $\Theta(n/r)$ operations. We next describe how to to handle edge insertions/deletions and then discuss the query operation.

\paragraph{Handling Edge Insertions/Deletions.} Edge insertions and deletions are handled in a similar way as in Section~\ref{sec: fullyPlanarEdge}, expect that here we will need to recompute a constant number of distance sparsifiers (rather than approximate Schur Complements) for each update. By Lemma~\ref{lemm: VertexDistanceSpars}, each distance sparsification can be done in time $\tilde{O}(r^{3/2}+|K|^2)$, where $K$ is the boundary vertex set of the corresponding region. Note that each update only increase by a constant the total number of regions and the total number of boundary vertices, and the size of each region after every update is at most $r$. Since the data structure will be recomputed from scratch every order $n/r$ operations, this will ensure that throughout the sequence of operations, the total number of regions is at most $O(n/r)$, and $\sum_i|\partial(P_i)|\leq O(n/\sqrt{r})$ and $\sum_i|\partial(P_i)|^2\leq r\cdot \sum_i|\partial(P_i)| =O(n\sqrt{r})$. Furthermore, at any time, a region can have at most $b:=O(\min\{r,\sqrt{r}+\frac{n}{r}\})$ boundary vertices. Thus, a call to \textsc{DistanceSparsify} will run in time $\tilde{O}(r^{3/2}+b^2)$.


Amortizing the initialization over $\Theta(n/r)$ operations gives that the expected amortized update time per operation is
\[
\tilde{O}\left(\frac{n^{1+\xi}+n\sqrt{r}}{n/r} + r^{3/2}+b^2\right)=\tilde{O}(n^{\xi}r+r^{3/2}+b^2).
\]
\paragraph{Handling Queries.} To answer the query \textsc{ShortestPath}$(s,t)$, similarly to Section~\ref{sec: fullyPlanarEdge}, we first form an auxiliary graph $H$ that is the union of the regions $P_s,P_t$ containing $s,t$, respectively, and the vertex distance sparsifiers of the remaining regions with respect to their boundary nodes. Then, we run a single-source shortest path algorithm from $s$ on top of $H$ and output $d_H(s,t)$. Since each $s-t$ shortest path that goes through other regions (different from $P_s,P_t$) must use the corresponding boundary vertices, and the boundary vertex pairwise distances are preserved within a factor $(2q-1)$ in $H$, we can guarantee that the output $d_H(s,t)$ is a $(2q-1)$-approximation to the $s-t$ shortest path in the current graph $G$, for any $q \geq 1$.


To bound the query time, we need to bound the size (or the number of edges) of $H$. Note that the size of the regions that contain $s$ and $t$ are bounded by $O(r)$. Note that throughout the algorithm we always ensure that each region has at most $b$ boundary nodes and its distance sparsifier has at most $O(qb^{1+1/q})$ edges by construction. Since there are at most $O(n/r)$ regions, it follows that the size of the union over the vertex distance sparsifiers is bounded by $\tilde{O}(\frac{n}{r} \cdot q b^{1+1/q})$. Thus, 
\[
	|E(H)| \leq \tilde{O}\left(\frac{n}{r} \cdot q b^{1+1/q} + r\right).
\]
Since the single-source shortest path algorithm runs in $\tilde{O}(|E(H)|)$ time, by choosing $r=n^{4/7}$, we get that $b=n^{3/7}$, the amortized update time is $\tilde{O}(n^{6/7})$, and the worst-case query time is $\tilde{O}(qn^{\frac{1}{7}(6 + 3/q)})$, for any $q \geq 1$. By using the global rebuilding technique, we can again guarantee the worst-case expected update time. This finishes the proof of Theorem~\ref{thm:shortest_path}.

\bibliography{literature}


\end{document}